\documentclass[11pt]{article} %
\usepackage{pgfplots}
\pgfplotsset{width=7cm,compat=1.9}
\usepackage{authblk}
\usepackage{booktabs}
\usepackage{float}
\usepackage[margin=1in]{geometry}
\usepackage[scr=rsfso]{mathalfa}
\usepackage{mdframed} %
\usepackage{microtype} %
\usepackage{mathpazo}
\usepackage{amsthm}
\usepackage{xspace}
\usepackage{comment}
\usepackage{lipsum}
\usepackage{amsfonts}
\usepackage{graphicx}
\usepackage{epstopdf}
\usepackage{algorithm}
\usepackage[noend]{algorithmic}
\usepackage{tikz}
\usepackage{caption}
\usepackage{subcaption}
\usepackage{amsmath}
\usepackage{hyperref}
\usepackage{amssymb}
\usepackage{dirtytalk}
\usetikzlibrary{decorations.pathmorphing}
\ifpdf
  \DeclareGraphicsExtensions{.eps,.pdf,.png,.jpg}
\else
  \DeclareGraphicsExtensions{.eps}
\fi

\let\epsilon=\varepsilon %

\newcommand{\ST}{\textsc{Steiner Tree}\xspace}
\newcommand{\CNC}{\textsc{Connected Vertex Cover}\xspace}
\newcommand{\SF}{\textsc{Steiner Forest}\xspace}
\newcommand{\DST}{\textsc{Directed Steiner Tree}\xspace}
\newcommand{\setcover}{\textsc{Set Cover}\xspace}
\newcommand{\FL}{\textsc{Facility Location}\xspace}
\newcommand{\GST}{\textsc{Group Steiner Tree}\xspace}
\newcommand{\PG}{\textsc{Projection Game Conjecture}\xspace}
\newcommand{\NWST}{\textsc{Node-Weighted Steiner Tree}\xspace}
\newcommand{\NWSF}{\textsc{Node-Weighted Steiner Forest}\xspace}
\newcommand{\NWPCST}{\textsc{Node-Weighted Prize Collecting Steiner Tree}\xspace}

\def\added{\mathop{\mathcal{L}}\nolimits}
\def\opt{\mathop{\rm OPT}\nolimits}

\def\antenna{\mathop{\rm Ant}\nolimits}
\def\killer{\mathop{\rm Killer}\nolimits}
\def\expansion{\mathop{\rm Exp}\nolimits}
\def\explevel{\mathop{\rm Elevel}\nolimits}
\def\auxGraph{\mathop{\rm H_{aux}}\nolimits}
\def\fs{\mathop{F_{l}}\nolimits}
\def\np{\mathop{\rm NP}\nolimits}
\def\p{\mathop{\rm P}\nolimits}
\def\dtime{\mathop{\rm DTIME}\nolimits}
\def\polylog{\mathop{\rm polylog}\nolimits}
\def\zptime{\mathop{\rm ZPTIME}\nolimits}

\newtheorem{theorem}{Theorem} % [section] %
\newtheorem{lemma}[theorem]{Lemma} %
\newtheorem{definition}[theorem]{Definition} %
\newtheorem{remark}[theorem]{Remark} %
\newtheorem{claim}[theorem]{Claim} %

\begin{document}

\title{A Constant-Factor Approximation for Quasi-bipartite Directed Steiner Tree on Minor-Free Graphs}

\author[]{Zachary Friggstad\thanks{Supported by an NSERC Discovery Grant and NSERC Discovery Accelerator Supplement Award.} }
\author[]{Ramin Mousavi}

\affil[]{Department of Computing Science, University of Alberta, Edmonton, Canada.
\authorcr
  \{\tt zacharyf@ualberta.ca, mousavih@ualberta.ca\}}

\date{}

\maketitle

\begin{abstract}
We give the first constant-factor approximation algorithm for quasi-bipartite instances of \DST on graphs that exclude fixed minors. In particular, for $K_r$-minor-free graphs our approximation guarantee is $O(r\cdot\sqrt{\log r})$ and, further, for planar graphs our approximation guarantee is 20.

Our algorithm uses the primal-dual scheme. We employ a more involved method of determining when to buy an edge while raising dual variables since, as we show, the natural primal-dual scheme fails to raise enough dual value to pay for the purchased solution. As a consequence, we also demonstrate integrality gap upper bounds on the standard cut-based linear programming relaxation for the \DST instances we consider.
\end{abstract}

\section{Introduction}
\sloppy
In the \DST (DST) problem, we are given a directed graph $G=(V,E)$ with edge costs $c(e)\geq 0$ for all $e\in E$, a root node $r\in V$, and a collection of terminals $X\subseteq V\setminus\{r\}$. The nodes in $V\setminus(X\cup \{r\})$ are called {\em Steiner} nodes. The goal is to find a minimum cost subset $F\subseteq E$ such that there is an $r-t$ path using only edges in $F$ for every terminal $t\in X$. Note any feasible solution that is inclusion-wise minimal must be an arborescence rooted at $r$. Throughout, we let $n$ denote $|V|$. 

One key aspect of DST lies in the fact that it generalizes many other important problems, e.g. \setcover, (non-metric, multilevel) \FL, and \GST. Halperin and Krauthgamer \cite{halperin2003polylogarithmic} showed \GST cannot be approximated within $O(\log^{2-\epsilon}n)$ for any $\epsilon>0$  unless $\np\subseteq\dtime{(n^{\polylog{(n)}})}$ and therefore the same result holds for DST. 

Building on a height-reduction technique of Calinescu and Zelikovsky \cite{calinescu,zelikovsky1997series}, Charikar et al. give the best approximation for DST which is an $O(|X|^\epsilon)$-approximation for any constant $\epsilon > 0$ \cite{charikar1999approximation} and also an $O(\log^3 |X|)$-approximation in $O(n^{{\rm polylog}(k)})$ time (quasi-polynomial time).
More recently, Grandoni, Laekhanukit, and Li \cite{grandoni2019log2} obtained a quasi-polynomial time $O(\frac{\log^2|X|}{\log\log |X|})$-approximation factor for \DST which is the best possible for quasi-polynomial time algorithms, assuming both the \PG and $\np\nsubseteq \bigcap_{0<\delta<1}\zptime(2^{n^\delta})$. Ghuge and Nagarajan \cite{ghuge2020quasi} studied a variant of DST called the \textsc{Directed Tree Orienteering} problem and presented an $O(\frac{\log |X|}{\log\log|X|})$-approximation in quasi-polynomial time which yields the same approximation guarantee as in \cite{grandoni2019log2}.

Methods based on linear programming have been less successful. Zosin and Khuller \cite{zosin2002directed} showed the integrality gap of a natural flow-based LP relaxation is $\Omega(\sqrt{|X|})$ but $n$, the number of vertices, in this example is exponential in terms of $|X|$. More recently, Li and Laekhanukit \cite{li2021polynomial} provided an example showing the integrality gap of this LP is at least polynomial in $n$.
On the positive side, \cite{rothvoss2011directed} shows for $\ell$-layered instances of DST that applying $O(\ell)$ rounds of the Lasserre hierarchy to a slight variant of the natural flow-based LP relaxation yields a relaxation with integrality gap $O(\ell \cdot \log |X|)$. This was extended to the LP-based Sherali-Adams and Lov\'{a}sz-Schrijver hierarchies by \cite{friggstad2014linear}.

% On the positive side, \cite{rothvoss2011directed} and \cite{friggstad2014linear} showed for some LP relaxations for DST that are similar to the standard flow-based LP relaxation, applying $O(\ell)$ rounds of some LP/SDP hierarchies such as Lasserre hierarchy, Sherali-Adams and Lov\'{a}sz-Schrijver, yields a relaxation with integrality gap $O(\ell\cdot\log |X|)$ for $\ell$-layered DST instances. 

We consider the cut-based relaxation \eqref{lp:primal} for DST, which is equivalent to the flow-based relaxation considered in \cite{zosin2002directed,li2021polynomial}; the flow-based relaxation is an extended formulation of \eqref{lp:primal}. Let $\delta^{in}(S)$ be the set of directed edges entering a set $S \subseteq V$,
\begin{alignat}{3}
\text{minimize:} & \quad & \sum_{e\in E} c(e) \cdot x_e \tag{{\bf Primal-LP}} \label{lp:primal} \\
\text{subject to:} && x(\delta^{in}(S)) \geq \quad & 1 \quad && \forall  S\subseteq V\setminus\{r\},~S\cap X\neq \emptyset \label{primal const} \\
&& x \geq \quad & 0 \notag
\end{alignat}

It is useful to note that if $|X|=1$ (the shortest $s-t$ path problem) or $X\cup\{r\}=V$ (the minimum cost arborescence problem), the extreme points of \eqref{lp:primal} are integral, see \cite{papadimitriou1998combinatorial} and \cite{edmonds1967optimum} respectively.

The undirected variant of \ST has seen more activity\footnote{One usually does not specify the root node in \ST, the goal is simply to ensure all terminals are connected.}. A series of papers steadily improved over the simple 2-approximation \cite{zelikovsky199311, karpinski1997new, promel2000new, robins2005tighter} culminating in a $\ln{4}+\epsilon$ for any constant $\epsilon>0$ \cite{byrka2013steiner}.
Bern and Plassmann \cite{bern1989steiner} showed that unless $\p=\np$ there is no approximation factor better than $\frac{96}{95}$ for \ST. However, there is a PTAS for \ST on planar graphs \cite{borradaile2009n} and more generally \cite{bateni2011approximation} obtains a PTAS for \SF on graphs of bounded-genus.

Another well-studied restriction of \ST is to quasi-bipartite graphs. These are the instances where no two Steiner nodes are connected by an edge (i.e., $V\setminus (X\cup \{r\})$ is an independent set). Quasi-bipartite instances were first studied by Rajagopalan and Vazirani \cite{rajagopalan1999bidirected} in order to study the bidirected-cut relaxation of the \ST problem: this is exactly \eqref{lp:primal} where we regard both directions of an undirected edge as separate entities. Feldmann et al. \cite{feldmann2016equivalence} studied \ST on graphs that do not have an edge-induced claw on Steiner vertices, i.e., no Steiner vertex with three Steiner neighbours, and presented a faster $\ln(4)$-approximation than the algorithm of \cite{byrka2013steiner}. Currently, the best approximation in quasi-bipartite instances of \ST is $\frac{73}{60}$-approximation \cite{goemans2012matroids}.

Naturally, researchers have considered quasi-bipartite instances of DST. Hibi and Fujito \cite{hibi2012multi} presented an $O(\log |X|)$-approximation algorithm for this case. Assuming $\p\neq\np$, this result asymptotically matches the lower bound $(1-o(1))\cdot\ln |X|$ for any $\epsilon>0$; this lower bound comes from the hardness of \setcover \cite{feige1998threshold, dinur2014analytical} and the fact that the quasi-bipartite DST problem generalizes the \setcover problem. Friggstad, K\"{o}nemann, and Shadravan \cite{friggstad2016logarithmic} showed that the integrality gap of \eqref{lp:primal} is also $O(\log |X|)$ by a primal-dual algorithm and again this matches the lower bound on the integrality gap of this LP up to a constant. 

More recently, Chan et al. \cite{chan2019polylogarithmic} studied the $k$-connected DST problem on quasi-bipartite instances in which the goal is to find a minimum cost subgraph $H$ such that there are $k$ edge-disjoint paths (in $H$) from $r$ to each terminal in $X$. They gave an upper bound of $O(\log |X|\cdot\log k)$ on the integrality gap of the standard cut-based LP (put $k$ instead of $1$ in the RHS of the constraints in \eqref{lp:primal}) by presenting a polynomial time randomized rounding algorithm.

It is worth noting that Demaine, Hajiaghayi, and Klein \cite{demaine2014node} show that if one takes a standard flow-based
relaxation for DST in planar graphs and further constraints
the flows to be ``non-crossing'', then the solution
can be rounded to a feasible DST solution while losing only a constant factor in the cost. To date, we do not know how to compute
a low-cost, non-crossing flow in polynomial time for DST instances on planar graphs.

% give an alternative relaxation for DST in planar graphs. However, this relaxation has non-linear constraints. They provided a rounding algorithm that turns a solution of such relaxation to a feasible integral solution for planar instances of DST while losing only a constant factor in the rounding. Therefore, if there is a polynomial time algorithm to solve this relaxation, then there is a constant factor approximation algorithm for DST instances on planar graphs: to date no such algorithm is known.

\subsection{Primal-Dual Approximations for Steiner Tree Problems}

%The primal-dual scheme has been successfully deployed.
Consider the \NWST (NWST) problem which is similar to undirected \ST except the weight function is on the Steiner vertices instead of edges and can also be viewed as a special case of DST. Guha et al. \cite{guha1999efficient} presented a primal-dual algorithm with guarantee of $O(\ln n)$ which is asymptotically tight since NWST also generalizes set cover. K\"onemann, Sadeghian, and Sanit\`a \cite{konemann2013lmp} give an $O(\log n)$-approximation via primal-dual framework for a generalization of NWST called \NWPCST\footnote{A key aspect of their algorithm is that it is also {\em Lagrangian multiplier preserving}.}.

Demaine, Hajiaghayi, and Klein \cite{demaine2014node} considered a generalization of NWST called \NWSF (NWSF) on planar graphs and using the generic primal-dual framework of Goemans and Williamson \cite{goemans1997primal} they showed a $6$-approximation and further they extended their result to minor-free graphs. Later Moldenhauer \cite{moldenhauer2013primal} simplified their analysis and showed an approximation guarantee of $3$ for NWSF on planar graphs. 

An interesting, non-standard use of the primal-dual scheme is in the work of Chakrabarty, Devanur, and Vazirani \cite{chakrabarty2011new} for undirected, quasi-bipartite instances of \ST. They introduce a new ``simplex-embedding'' LP relaxation and their primal-dual scheme raises dual variables with different rates. It is worth noting that although they also obtain upper bound for the integrality gap of the so-called {\em bidirected cut relaxation} (BCR) of quasi-bipartite instances of \ST, the algorithm and the simplex-embedding LP relaxation itself are valid only in the undirected setting.

%To the best of our knowledge there is no constant factor approximation algorithm for planar instances of DST. 

\subsection{Our contributions}
%To the best of our knowledge, there is no concrete result about DST on planar graphs other than the results for the general instances of DST.
We present the first concrete result for approximating DST on planar graphs beyond what was known in general graphs. Namely, we present a primal-dual algorithm for DST on quasi-bipartite, minor-free graphs. 

Generally, it is difficult to effectively utilize primal-dual algorithms in directed network design problems. This is true in our setting as well: we begin by showing a standard primal-dual algorithm (similar to the primal-dual algorithm for the minimum-cost arborescence problem) does not grow sufficiently-large dual to pay for the set of edges it purchases within any constant factor.

We overcome this difficulty by highlighting different roles for edges in connecting the terminals to the root. For some edges, we maintain two slacks: while raising dual variables these two slacks for an edge may be filled at different rates (depending on the edge's role for the various dual variables being raised) and we purchase the edge when one of its slacks is exhausted. Furthermore, unlike the analysis of standard primal-dual algorithms where the charging scheme is usually more local (i.e., charging the cost of purchased edges to the dual variables that are \say{close by}), we need to employ a more global charging scheme. Our approach also provides an $O(1)$ upper bound on the integrality gap of the natural cut-based relaxation \eqref{lp:primal} for graphs that exclude a fixed minor.

%As we stated earlier, we provide a primal-dual algorithm for DST on quasi-bipartite, minor-free graphs.

%The only primal-dual algorithm known in DST problem is one for the special case when there is no Steiner node (i.e., the arborescence problem) \cite{edmonds1967optimum} and the other one is the work of \cite{friggstad2016logarithmic} on quasi-bipartite instances where they provide $O(\log |X|)$-approximation.
%However, in undirected settings researchers have utilized primal-dual algorithms more often.

%Generally, it is difficult to effectively utilize primal-dual algorithms in directed network design problems but there are some successful cases: the minimum cost arborescence algorithm in \cite{edmonds1967optimum} is one such example and, as mentioned above, \cite{friggstad2016logarithmic} give logarithmic integrality gap bound for quasi-bipartite DST.

%We take the first step towards a concrete result for planar instances of DST, a constant factor approximation algorithm for quasi-bipartite graphs that exclude a fixed minor. Furthermore, our approach yields an $O(1)$ upper bound on the integrality gap of the cut-based \eqref{lp:primal} for such graphs.

We summarize our results here.
\begin{theorem}\label{main thm minor free graphs}
There is an $O(r\cdot\sqrt{\log r})$-approximation algorithm for \DST on quasi-bipartite, $K_r$-minor free graphs. Moreover, the algorithm gives an upper bound of $O(r\cdot\sqrt{\log r})$ on the integrality gap of \eqref{lp:primal} for DST instances on such graphs.
\end{theorem}

\begin{remark}
The running time of our algorithm is $O(|V|^c)$ where $c$ is a fixed constant that is independent of $r$. Also, we only require that every (simple) minor of the graph has bounded average degree to establish our approximation guarantee. In particular, if every minor of the input (quasi-bipartite) graph has degree at most $d$, then the approximation factor will be $O(d)$.
\end{remark}

%We prove the above theorem by presenting a non-standard primal-dual algorithm. In Section \ref{standard primal-dual fails} we show that a standard primal-dual algorithm fails to grow enough value for the dual LP to be able to charge the bought cost within a constant factor of the total growth of the dual variables.

%See Remark \ref{remark size minor free} regarding the constant suppressed by the $O()$ notation in Theorem \ref{main thm minor free graphs}.

\begin{theorem}\label{main thm planar graphs}
There is a $20$-approximation algorithm for \DST on quasi-bipartite, planar graphs. Moreover, the algorithm gives an upper bound of $20$ on the integrality gap of \eqref{lp:primal} for \DST instances on such graphs. 
\end{theorem}

We also verify that \ST (and, thus, \DST) remains $\np$-hard even when restricted to quasi-bipartite, planar instances. Similar results are known, but we prove this one explicitly since we were not able to find this precise hardness statement in any previous work.
%We could not find any hardness result for \ST on quasi-bipartite planar graphs in the literature. We show the following simple fact.

\begin{theorem}\label{hardness}
\ST instances on bipartite planar graphs where the terminals are on one side and the Steiner nodes are on the other side is $\np$-hard.
\end{theorem}

The above hardness result shows DST instances on quasi-bipartite, planar graphs is $\np$-hard as well.

\subsection{Organization of the paper}
In Section \ref{preliminaries}, we state some definition and notation where we use throughout the paper. In Section \ref{standard primal-dual fails} we present an example that shows the most natural primal-dual algorithm fails to prove our approximation results, this helps the reader understand the key difficulty we need to overcome to make a primal-dual algorithm work and motivates our more refined approach. In Section \ref{the algorithm} we present our primal-dual algorithm and in Section \ref{the analysis} we present the analysis. The analysis contains three main subsections where in each section we present a charging scheme. The first two charging schemes are straightforward but the last one requires some novelty. Finally, we put all these charging schemes together in Subsection \ref{putting everything together} and prove Theorems \ref{main thm minor free graphs} \& \ref{main thm planar graphs}. Finally, in Section \ref{np hardness} we show the hardness result (Theorem \ref{hardness}).

\section{Preliminaries}\label{preliminaries}
In this paper, graphs are simple directed graphs unless stated otherwise. By simple we mean there are no parallel edges\footnote{Two edges are parallel if their endpoints are the same and have the same orientation.}. Note that we can simply keep the cheapest edge in a group of parallel edges if the input graph is not simple; the optimal value for DST problem does not change.

Throughout this paper, we fix a directed graph $G=(V,E)$, a cost function $c:E\to\mathbb{R}_{\geq 0}$, a root $r$, a set of terminals $X\subseteq V\setminus \{r\}$, and no edge between any two Steiner nodes, as the input to the DST problem. We denote the optimal value for DST instance by $\opt$.

Given a subgraph $G'$ of $G$ we define $\delta^{in}_{G'}(S)=\{e=(u,v)\in E(G'):~u\in V\setminus S,~v\in S\}$ (i.e., the set of edges in $G'$ entering $S$) we might drop the subscript if the underlying subgraph is $G$ itself. For an edge $e=(u,v)$, we call $u$ the \textbf{tail} and $v$ the \textbf{head} of $e$. By a \textbf{dipath} we mean a directed path in the graph. By \textbf{SCCs} of $F\subseteq E$ we mean the strongly connected components of $(V,F)$ that contains either the root node or at least one terminal node. So for example, if a Steiner node is a singleton strongly connected component of $(V,F)$ then we do not refer to it as an SCC of $F$. Due to the quasi-bipartite property, these are the only possible strongly connected components in the traditional sense of $(V,F)$ that we will not call SCCs. Observe $F$ is a feasible DST solution if and only if each SCC is reachable from $r$.
%{\bf ZF: Articulate that you mean only the SCCs that have a terminal or root.}

An {\em arborescence} $T=(V,E)$ rooted at $r\in V$ is a directed tree oriented away from the root such that every vertex in $V$ is reachable from $r$. %The following are some tree terminologies.
By {\em height} of a vertex $u$ in $T$ we mean the number of edges between $r$ (the root) and $u$ in the dipath from $r$ to $u$ in $T$.
%For two vertices $u$ and $v$, where $u$ is an ancestor of $v$, we denote by $P^T_{u,v}$ the path from $u$ to $v$ in $T$ (the superscript $T$ is dropped when $T$ is clear from the context).
We let $T_u$ denotes the subtree of $T$ rooted at $u$.

Our discussions, algorithm, and the analysis rely on the concept of {\em active sets}, so we define them here.

\begin{definition}[Violated set]\label{violated set}
Given a DST instance and a subset $F\subseteq E$, we say $S\subseteq V\setminus\{r\}$ where $S\cap X\neq \emptyset$ is a {\em violated} set with respect to $F$ if $\delta^{in}_F(S)=\emptyset$.
\end{definition}

\begin{definition}[Active set]\label{active moats}
Given a DST instance and a subset $F\subseteq E$, we call a minimal violated set (no proper subset of it, is violated) an {\em active} set (or active moat) with respect to $F$.
\end{definition}

We use the following definition throughout our analysis and (implicitly) in the algorithm.

\begin{definition}[$F$-path]\label{F_i-path}
We say a dipath $P$ is a $F$-path if all the edges of $P$ belong to $F\subseteq E$. We say there is a $F$-path from a subset of vertices to another if there is a $F$-path from a vertex of the first set to a vertex of the second set.
\end{definition}

In quasi-biparitite graphs, active moat have a rather \say{simple} structure, our algorithm will leverage the following properties.
%{\bf ZF: Make this a lemma? Do SCCs need to have a terminal?}
\begin{lemma}\label{remark on active moats}
Consider a subset of edges $F$ and let $A$ be an active set with respect to $F$. Then, $A$ consists of exactly one SCC $C_A$ of $F$, and any remaining in $A\setminus C_A$ are Steiner nodes. Furthermore, for every Steiner node in $A\setminus C_A$ there are edges in $F$ that are oriented from the Steiner node to $C_A$.
\end{lemma}
\begin{proof}
By definition of violated sets, $A$ does not contain $r$. If $A$ contains only one terminal, then the first statement holds trivially. So consider two terminals $t$ and $t'$ in $A$. We show there is a $F$-path from $t$ to $t'$ and vice versa. Suppose not and wlog assume there is no $F$-path from $t'$ to $t$. Let $B:=\{v\in A:~\exists F-path~from~v~to~t\}$. Note that $B$ is a violated set and $B\subseteq A\setminus\{t'\}$ which violates the fact that $A$ is a minimal violated set. Therefore, exactly one SCC of $F$ is in $A$.

Next we prove the second statement. Let $s$ be a Steiner node (if exists) in $A\setminus C_A$. If there is no edge in $F$ oriented from $s$ to $C_A$, then $A\setminus\{s\}$ is a violated set, because the graph is quasi-bipartite and the fact that $A$ is a violated set itself, contradicting the fact that $A$ is a minimal violated set.
\end{proof}

Note that the above lemma limits the interaction between two active moats. More precisely, two active moats can only share Steiner nodes that lie outside of the SCCs in the moats.

\begin{definition}[The SCC part of active moats]\label{SCC part of active moats}
Given a set of edges $F$ and an active set $A$ (with respect to $F$), we denote by $C_A$ the SCC (with respect to $F$) inside $A$.
\end{definition}
We use $C_A$ rather than $C_A^F$ because the set $F$ will always be clear from the context.

Finally we recall bounds on the size of $K_r$-minor free graphs that we use at the end of our analysis.

\begin{theorem}[Thomason 2001 \cite{thomason2001extremal}]\label{size of minor free graphs}
Let $G=(V,E)$ be a $K_r$-minor free graph with no parallel edges. Then, $|E|\leq O(r\cdot\sqrt{\log r})|V|$ and this bound is asymptotically tight.
\end{theorem}

\begin{remark}\label{remark size minor free}
We are not aware of the constant suppressed by the $O(.)$ notation in Thomason's result (Theorem \ref{size of minor free graphs}). But there is another result by Mader \cite{mader1968} that gives an upper bound of $8\cdot r\cdot\log r$ which asymptotically is worse than Thomason's result but the constant in the $O(.)$ is known. If we use this result in our analysis, we have a $2\cdot(8\cdot r\cdot\log r+1)$-approximation algorithm for DST instances on quasi-bipartite $K_r$-minor free graphs.
\end{remark}

Bipartite planar graphs are $K_5$-minor free, but we know of explicit bounds sizes. The following is the consequence of Euler's formula that will be useful in our tighter analysis for quasi-bipartite, planar graphs.
\begin{lemma}\label{size of bipartite planar}
Let $G=(V,E)$ be a bipartite planar graph with no parallel edges. Then, $|E|\leq 2\cdot|V|$.
\end{lemma}

\section{Standard primal-dual algorithm and a bad example}\label{standard primal-dual fails}

Given a DST instance with $G=(V,E)$, $r\in V$ as the root, and $X\subseteq V-\{r\}$ as the terminal set, we define $\mathcal{S}:=\{S\subsetneq V:r\notin S,~and~S\cap X\neq\emptyset\}$. We consider the dual of \eqref{lp:primal}.

\begin{alignat}{3}\label{dual lp}
\text{maximize:} & \quad & \sum_{S\in\mathcal{S}} y_S \tag{{\bf Dual-LP}} \\
\text{subject to:} && \sum_{\substack{S\in \mathcal{S}:\\ e\in\delta^{in}(S)}}y_S \leq \quad & c(e) \quad && \forall e \in E \label{cnst:cut} \\
&& y \geq \quad & 0 \notag
\end{alignat}

As we discussed in the introduction, a standard primal-dual algorithm solves arborescence problem on any directed graph \cite{edmonds1967optimum}. Naturally, our starting point was to investigate this primal-dual algorithm for DST instances. We briefly explain this algorithm here. At the beginning we let $F:=\emptyset$. Uniformly increase the dual constraints corresponding to active moats and if a dual constraint goes tight, we add the corresponding edge to $F$. Update the active sets based on $F$ (see Definition \ref{active moats}) and repeat this procedure. At the end, we do a reverse delete, i.e., we go over the edges in $F$ in the reverse order they have been added to $F$ and remove it if the feasibility is preserved. Unfortunately, for DST instances in quasi-bipartite planar graphs, there is a bad example (see Figure \ref{bad example}), that shows the total growth of the dual variables is $2+(2\cdot k+2)\cdot\epsilon$ while the optimal solution costs $k+1+(k+2)\cdot\epsilon$ for arbitrarily large $k$. So the dual objective is not enough to pay for the cost of the edges in $F$ (i.e., we have to multiply the dual objective by $O(k)$ to be able to pay for the edges in $F$).

\textbf{What is the issue and how can we fix it?} One way to get an $O(1)$-approximation is to ensure at each iteration the number of edges (in the final solution) whose dual constraints are losing slack at this iteration is proportioned to the number of active moats. In the bad example (Figure \ref{bad example}), when the bottom moat is paying toward the downward blue edges, there are only two active moats but there are $k$ downward blue edges that are currently being paid for by the growing dual variables.

To avoid this issue, we consider the following idea: once the bottom active moat grew enough so that the dual constraints corresponding to all the downward blue edges are tight we purchase an arbitrary one of them, say $(r,z_k)$ for our discussion here. Once the top active moat reaches $z_1$ instead of skipping the payment for this edge (since the dual constraint for $(w_2,z_1)$ is tight), we let the active moat pay towards this edge again by ignoring previous payments to the edge, and then we purchase it once it goes tight. Note that now we violated the dual constraint for $(w_2,z_1)$ by a multiplicative factor of $2$. Do the same for all the other downward blue edges (except $(r,z_k)$ that was purchased by the bottom moat). Now it is easy to see that we grew enough dual objective to approximately pay for the edges that we purchased. We make this notion precise by defining different roles for downward blue edges in the next section. In general, each edge can serve up to two roles and has two ``buckets'' in which it receives payment: each moat pays towards the appropriate bucket depending on the role that edge serves for that moat. An edge is only purchased if one of its buckets is filled and some tiebreaking criteria we mention below is satisfied.

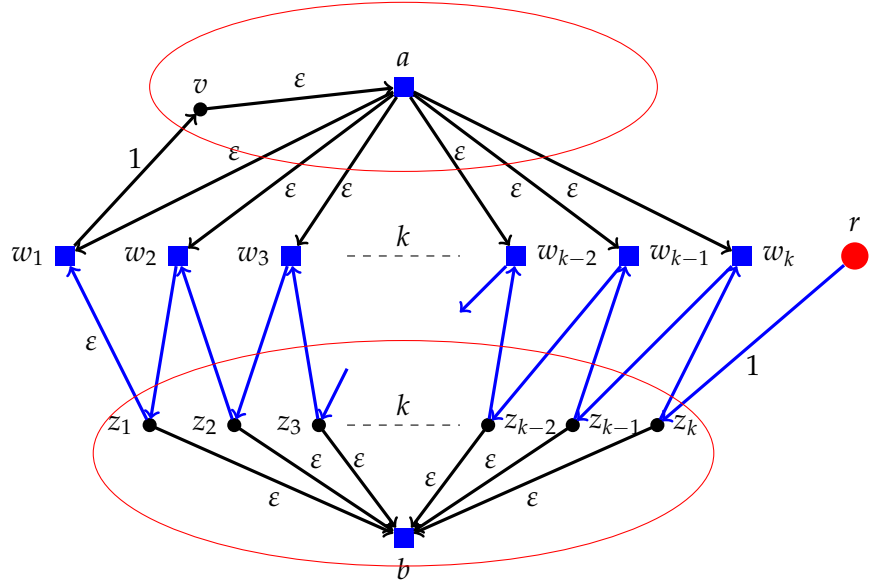
\begin{figure}
\centering
\begin{tikzpicture}[smallnode/.style={circle, draw=black, fill=black, very thick, scale =0.4},scale=.75]
%terminals
\node[shape=rectangle,fill=blue,label=above:$a$] (a) at (0,0) {};
\node[shape=circle,fill=red,label=above:$r$] (r) at (8,-3) {};
\node[shape=rectangle,fill=blue,label=below:$b$] (b) at (0,-8) {};

\node[smallnode,label=above:$v$] (v8) at (-3.6,-.4) {};

\node[shape=rectangle,fill=blue, label=left:$w_1$] (u1) at (-6,-3) {};
\node[shape=rectangle,fill=blue, label=left:$w_2$] (u2) at (-4,-3) {};
\node[shape=rectangle,fill=blue, label=left:$w_3$] (u3) at (-2,-3) {};
%\node[shape=rectangle,fill=blue] (u4) at (0,-3) {};
\node[shape=rectangle,fill=blue, label=right:$w_{k-2}$] (u5) at (2,-3) {};
\node[shape=rectangle,fill=blue, label=right:$w_{k-1}$] (u6) at (4,-3) {};
\node[shape=rectangle,fill=blue, label=right:$w_k$] (u7) at (6,-3) {};

\draw[->, very thick] (a) -- (u1) node[midway, above] {$\epsilon$};
\draw[->, very thick] (a) -- (u2) node[midway, below] {$\epsilon$};
\draw[->, very thick] (a) -- (u3) node[midway, below] {$\epsilon$};
%\draw[->, very thick] (a) -- (u4) node[midway, right] {$\epsilon$};
\draw[->, very thick] (a) -- (u5) node[midway, above] {$\epsilon$};
\draw[->, very thick] (a) -- (u6) node[midway, below] {$\epsilon$};
\draw[->, very thick] (a) -- (u7) node[midway, below] {$\epsilon$};
\draw[->, very thick] (u1) -- (v8) node[midway, above] {$1$};
\draw[->, very thick] (v8) -- (a) node[midway, above] {$\epsilon$};

%steiner nodes
\node[smallnode, label=left:$z_1$] (v1) at (-4.5,-6) {};
\node[smallnode, label=left:$z_2$] (v2) at (-3,-6) {};
\node[smallnode, label=left:$z_3$] (v3) at (-1.5,-6) {};
%\node[smallnode] (v4) at (0,-6) {};
\node[smallnode, label=right:$z_{k-2}$] (v5) at (1.5,-6) {};
\node[smallnode, label=right:$z_{k-1}$] (v6) at (3,-6) {};
\node[smallnode, label=right:$z_{k}$] (v7) at (4.5,-6) {};

\draw[dashed] (-1,-3) -- (1,-3) node[midway,above] {$k$};
\draw[dashed] (-1,-6) -- (1,-6) node[midway,above] {$k$};

\draw[->, very thick,blue] (r) -- (v7) node[black, midway, below] {$1$};
\draw[->, very thick,blue] (v7) -- (u7) node[midway, left] {};
\draw[->, very thick,blue] (u7) -- (v6) node[midway, below] {};
\draw[->, very thick,blue] (v6) -- (u6) node[midway, below] {};
\draw[->, very thick,blue] (u6) -- (v5) node[midway, left] {};
%\draw[->, very thick,blue] (v5) -- (v5) node[midway, below] {};
\draw[->, very thick,blue] (v5) -- (u5) node[midway, below] {};
\draw[->, very thick,blue] (u5) -- (1,-4) {};
\draw[->, very thick,blue] (-1,-5) -- (v3) {};
%\draw[->, very thick,blue] (-1,-4) -- (u3) {};
\draw[->, very thick,blue] (v3) -- (u3) node[midway, left] {};
\draw[->, very thick,blue] (u3) -- (v2) node[midway, below] {};
\draw[->, very thick,blue] (v2) -- (u2) node[midway, left] {};
\draw[->, very thick,blue] (u2) -- (v1) node[midway, below] {};
\draw[->, very thick,blue] (v1) -- (u1) node[black,midway, left] {$\epsilon$};

\draw[->, very thick] (v1) -- (b) node[midway, below] {$\epsilon$};
\draw[->, very thick] (v2) -- (b) node[midway, above] {$\epsilon$};
\draw[->, very thick] (v3) -- (b) node[midway, above] {$\epsilon$};
%\draw[->, very thick] (v4) -- (b) node[midway, right] {$\epsilon$};
\draw[->, very thick] (v5) -- (b) node[midway, left] {$\epsilon$};
\draw[->, very thick] (v6) -- (b) node[midway, above] {$\epsilon$};
\draw[->, very thick] (v7) -- (b) node[midway, below] {$\epsilon$};

\draw[red] (a) ellipse (4.5 and 1.5);
\draw[red] (0,-6.5) ellipse (5.5 and 2);
\end{tikzpicture}
\caption{This is an example to show why a standard primal-dual algorithm fails. The square vertices are terminals. The downward blue edges (i.e., $(w_{i},z_{i-1})$'s for $2\leq i\leq k$) have cost $1$, the upward blue edges (i.e., $(z_{i},w_{i})$'s for $1\leq i\leq k$) have cost $\epsilon$. The cost of the black edges are shown in the picture. Note any feasible solution contains all the blue edges and the cost of an optimal solution is $k+1+(k+2)\cdot\epsilon$. However, it is easy to see the total dual variables that are grown using a standard primal-dual algorithm is $2+(2\cdot k+2)\cdot\epsilon$.}
\label{bad example}
\end{figure}

\section{Our primal-dual algorithm}\label{the algorithm}
As we discussed in the last section, we let the algorithm violate the dual constraint corresponding to an edge by a factor of $2$ and hence we work with the following modified Dual-LP:

\begin{alignat}{3}\label{dual const modified}
\text{maximize:} & \quad & \sum_{S\in\mathcal{S}} y_S \tag{{\bf Dual-LP-Modified}}  \\
\text{subject to:} && \sum_{\substack{S\in \mathcal{S}:\\ e\in\delta^{in}(S)}}y_S \leq \quad & 2\cdot c(e) \quad && \forall e \in E \label{cnst:cut modified} \\
&& y \geq \quad & 0 \notag
\end{alignat}

Note that the optimal value of \eqref{dual const modified} is at most twice the optimal value of \eqref{dual lp} because consider a feasible solution $y$ for the former LP then $\frac{y}{2}$ is feasible for the latter LP.

Let us define the different buckets for each edge that are required for our algorithm.

\noindent
\textbf{Antenna, expansion and killer buckets:}\\
We say edge $e=(u,v)$ is an {\em antenna} edge if $u\notin X\cup \{r\}$ and $v\in X$, in other words, if the tail of $e$ is a Steiner node and the head of $e$ is a terminal. For every antenna edge we associate an antenna bucket with size $c(e)$. For every non-antenna edge $e$, we associate two buckets, namely {\em expansion} and {\em killer} buckets, each of size $c(e)$. The semantics of these labels will be introduced below.

Now we, informally, describe our algorithm, see Algorithm \ref{our alg} for the detailed description. Recall the definition of active moats (Definition \ref{active moats}).

\noindent
\textbf{Growth phase:}
At the beginning of the algorithm we set $F:=\emptyset$ and every singleton terminal is an active moat. As long as there is an active moat with respect to $F$ do the following: uniformly increase the dual variables corresponding to the active moats. Let $e\notin F$ be an antenna edge with its head in an active moat, then the active moat pays towards the antenna bucket of $e$. Now suppose $e=(u,v)\notin F$ is a non-antenna edge, so $u\in X\cup\{r\}$. For every active moat $A$ that contains $v$, if $C_A$\footnote{See Definition \ref{SCC part of active moats}} is a subset of an active set $A'$ with respect to $F\cup\{e\}$, then $A$ pays toward the expansion bucket of $e$ and otherwise $A$ pays towards the killer bucket of $e$.

Uniformly increase the dual variables corresponding to active moats until a bucket for an edge $e$ becomes full (antenna bucket in case $e$ is an antenna edge, and expansion or killer bucket if $e$ is a non-antenna edge), add $e$ to $F$. Update the set of active moats $\mathcal{A}$ according to set $F$.

\noindent
\textbf{Pruning:}
Finally, we do the standard reverse delete meaning we go over the edges in $F$ in the reverse order they have been added and if the resulting subgraph after removing an edge is still feasible for the DST instance, remove the edge and continue.

The following formalizes the different roles of a non-antenna edge that we discussed above.

\begin{definition}[Relation between non-antenna edges and active moats]\label{relation between edges and active moats}
Given a subset of edges $F\subseteq E$, let $\mathcal{A}$ be the set of all active moats with respect to $F$. Consider a non-antenna edge $e=(u,v)$ (so $u\in X\cup\{r\}$). Suppose $v\in A$ where $A\in\mathcal{A}$. Then,
\begin{itemize}
    \item we say $e$ is an expansion edge with respect to $A$ under $F$ if there is a subset of vertices $A'$ that is active with respect to $F\cup\{e\}$ such that $C_A\subsetneq A'$,
    \item otherwise we say $e$ is a killer edge with respect to $A$.
\end{itemize}
\end{definition}
For example, all exiting edges from $r$ that are not in $F$ is a killer edge with respect to any active moat (under $F$) it enters. See Figure \ref{pic for relation between edges and moats} for an illustration of the above definition.

{\bf Intuition behind this definition:} When $e=(u,v)$ is a killer edge with respect to an active moat $A$, then there is a dipath in $F\cup\{e\}$ from $r$ or $C_{A'}$ to $C_A$ where $A'\neq A$ is an active moat with respect to $F$. Note that adding $e$ to $F$ will make the dual variable corresponding to $A$ stop growing and that is why we call $e$ a killer edge with respect to $A$. For example, in Figure \ref{pic for relation between edges and moats}, both $e$ and $e'$ are killer edges with respect to $A'$. On the other hand, if $e=(u,v)$ is an expansion edge with respect to $A$, then $C_A$ will be a part of a \say{bigger} active moat with respect to $F\cup\{e\}$ and hence the name expansion edge for $e$. For example, in Figure \ref{pic for relation between edges and moats}, $e$ is an expansion edge with respect to $A$ because in $F\cup\{e\}$, $A\cup B\cup\{s\}$ is an active moat whose SCC contains $C_A$.

\begin{figure}
\centering
\begin{tikzpicture}[smallnode/.style={circle, draw=black, fill=black, very thick, scale =0.4},
bignode/.style={circle, draw=black, very thick, scale =1.5},scale=.75]

\node[bignode, red, label=left:$B$] (b1) at (0,.5){};
\node[bignode, blue] (ca1) at (-1,-2){};
\node[bignode, blue] (ca2) at (1,-2){};
\node[bignode, red, label=left:$B'$] (b2) at (0.5,-4){};
\node[bignode, blue] (ca3) at (3.5,-3){};
\node[smallnode, label=right:$s$] (s) at (0,-1){};
\node[smallnode, label=below:$s'$] (s') at (2,-3){};

\draw[->, blue, very thick] (b1) -- (s) node[midway, right] {$e$};
\draw[->, very thick] (s) -- (ca1);
\draw[->, very thick] (s) -- (ca2);
\draw[->, blue, very thick] (ca2) to [out=90,in=35] node[pos=.5, right] {$e''$} (s);
\draw[->, very thick] (s') -- (ca2);
\draw[->, very thick] (s') -- (ca3);
\draw[->, blue, very thick] (b2) -- (s') node[pos=.4, above] {$e'$};
\draw[->, decorate, decoration=zigzag] (ca2) -- (b2) node [midway, left] {$\fs$};
\draw[->, decorate, decoration=zigzag] (ca1) -- (b1) node [midway, left] {$\fs$};
\draw[->, decorate, decoration=zigzag] (ca1) -- (b2) node [midway, left] {$\fs$};

\draw[blue, rotate=-45] (.7, -2) ellipse (.5 and 1.5);
\draw[blue, rotate=45] (-.7, -2) ellipse (.5 and 1.8);
\draw[blue] (3, -3) ellipse (1.4 and .5);

\node[label=left:$A$] at (-1.5,-2) {}; 
\node[label=right:$A'$] at (1.5,-1.7) {}; 
\node[label=right:$A''$] at (4.3,-3) {}; 

\end{tikzpicture}
\caption{Above is a part of a graph at the beginning of iteration $l$ in the algorithm. $\fs$ denotes the set $F$ at this iteration. The circles are SCCs in $(V,\fs)$. Blue circles are inside some active moats shown with ellipses. The black dots $s$ and $s'$ are Steiner nodes. The black edges and the zigzag paths are in $\fs$. The edges $e,e'$, and $e''$ have not been purchased yet (i.e., $e,e',e''\notin \fs$). Since $C_A$ is a subset of an active moat namely $A\cup B\cup\{s\}$ with respect to $\fs\cup\{e\}$, $e$ is an expansion edge with respect to $A$. However, $e$ is a killer edge with respect to $A'$ and $e''$ is a killer edge with respect to $A$. Finally, $e'$ is a killer edge with respect to $A'$ (and $A''$) because there is a $\fs\cup\{e'\}$-path from $C_A$ to $C_{A'}$ (and $C_{A''}$), therefore $C_{A'}$ (and $C_{A''}$) cannot be inside an active moat with respect to $\fs\cup\{e'\}$.}
\label{pic for relation between edges and moats}
\end{figure}
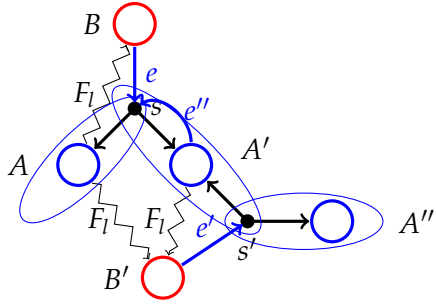

Now we can state our algorithm in detail, see Algorithm \ref{our alg}. Note that the purchased edge $e_l$ at iteration $l$ enters some active moat at iteration $l$. 

After the algorithm finishes, then we label non-antenna edges by expansion/killer as determined by the following rule:
\begin{definition}[Killer and expansion edges]\label{def killer and expansion edge}
Consider iteration $l$ of the algorithm where we added a non-antenna edge $e_l$ to $F$. We label $e_l$ as expansion (killer) if the expansion (killer) bucket of $e$ becomes full at iteration $l$, break ties arbitrarily. 
\end{definition}
Following remark helps to understand the above definition better.
\begin{remark}
It is possible that one bucket becomes full for an edge yet we do not purchase the edge with that bucket label (killer or expansion) due to tiebreaking when multiple buckets become full. For example, this would happen in our bad example for the downward blue edges: their killer buckets are full yet all but one are purchased as expansion edges.
\end{remark}

Let us explain the growth phase of Algorithm \ref{our alg} on the bad example in Figure \ref{bad example}. Since the early iterations of the algorithm on this example are straightforward, we start our explanation from the iteration where the active moats are $A=\{b,z_1,z_2,...,z_k\}$ and $A'=\{a,v\}$. 

Every $(w_i,z_{i-1})$ for $2\leq i\leq k$ is a killer edge with respect to $A$ so $A$ pays toward the killer buckets of these edges. At the same iteration, $(w_1,v)$ is an expansion edge with respect to $A'$ so $A'$ pays toward the expansion bucket of this edge. Now the respected buckets for all mentioned edges are full. Arbitrarily, we pick one of these edges, let us say $(w_k,z_{k-1})$, and add it to $F$. Then, $A$ stops growing. In the next iteration, we only have one active moat $A'$. Since $(w_1,v)$ is still expansion edge with respect to $A'$ and its (expansion) bucket is full, in this iteration we add $(w_1,v)$ to $F$ and after updating the active moats, again we only have one active moat $\{a,v,w_1\}$ which by abuse of notation we denote it by $A'$. Next iteration we buy the antenna edge $(w_1,z_1)$ and the active moat now is $A'=\{a,v,w_1,z_1\}$. In the next iteration, the crucial observation is that the killer bucket of $(w_2,z_1)$ is full (recall the $A$ payed toward the killer bucket of $(w_2,z_1)$); however, $(w_2,z_1)$ is an expansion edge with respect to $A'$ so $A'$ will pay towards its expansion bucket and then purchases it. Similarly, the algorithm buys $(w_i,z_{i-1})$'s except $(w_k,z_k)$ because this edge is in $F$ already (recall we bought this edge with $A$). Finally, $(r,z_k)$ is a killer edge with respect to the active moat in the last iteration and we purchase it.

\begin{algorithm}
\caption{Primal-Dual Algorithm for DST on Quasi-Bipartite Graphs}
{\bf Input}: Directed quasi-bipartite graph $G = (V,E)$ with edge costs $c(e) \geq 0$ for $e \in E$, a set of terminal $X\subseteq V\setminus \emptyset$, and a root vertex $r$. \\
{\bf Output}: An arborescence $\overline F$ rooted at $r$ such that each\\ terminal is reachable from $r$ in $\overline F$.
\begin{algorithmic}
\STATE $\mathcal{A} \leftarrow \{\{v\} : v \in X\}$. \COMMENT{The active moats each iteration, initially all singleton terminal set.}
\STATE $y^* \gets 0$. \COMMENT{The dual solution}
\STATE $ F \gets \emptyset$. \COMMENT{The edges purchased}
\STATE $l \gets 0$. \COMMENT{The iteration counter}
\STATE $b^{\antenna}_e\gets 0$, $b^{\expansion}_e\gets 0$ and $b^{\killer}_e\gets 0$. \COMMENT{The buckets}
\STATE \textbf{Growing phase:}
\WHILE{until $\mathcal{A}\neq\emptyset$}
\STATE Find the maximum value $\epsilon\geq 0$ such that the following holds:
\STATE ~~~(a) for every antenna edge $e$ we have $b^{\antenna}_e+\sum\limits_{\substack{A\in\mathcal{A}:\\e\in\delta^{in}(A)}}\epsilon\leq c(e)$.
\STATE ~~~(b) for every non-antenna edge $e$ we have $b^{\expansion}_e+\sum\limits_{\substack{A\in \mathcal{A}:\\e~is~expansion\\with~resp.~to~A}}\epsilon\leq c(e)$.
\STATE ~~~(c) for every non-antenna edge $e$ we have $b^{\killer}_e+\sum\limits_{\substack{A\in \mathcal{A}:\\e~is~killer~with\\resp.~to~A}}\epsilon\leq c(e)$.

\STATE Increase the dual variables $y^*$ corresponding to each active moat by $\epsilon$.
\FOR{every antenna edge $e$} 
\STATE $b^{\antenna}_e\gets b^{\antenna}_e+\sum\limits_{\substack{A\in\mathcal{A}:\\e\in\delta^{in}(A)}}\epsilon$.
\ENDFOR
\FOR{every non-antenna edge $e$} 
\STATE $b^{\expansion}_e\gets b^{\expansion}_e+\sum\limits_{\substack{A\in \mathcal{A}:\\e~is~expansion\\with~resp.~to~A}}\epsilon$.
\STATE $b^{\killer}_e\gets b^{\killer}_e+\sum\limits_{\substack{A\in \mathcal{A}:\\e~is~killer~with\\resp.~to~A}}\epsilon$.
\ENDFOR
\STATE pick any single edge $e_l \in \cup_{A \in \mathcal{A}} \delta^{in}(A)$ with one of (a)-(c) being tight (break ties arbitrarily).
\STATE $F\gets F\cup \{e_l\}$.
\STATE update $\mathcal{A}$ based on the minimal violated sets with respect to $F$.
\STATE $l\gets l+1$. 
\ENDWHILE
\STATE \textbf{Deletion phase:}
\STATE $\overline F\gets F$.
\FOR{$i$ from $l$ to $0$}
\IF{$\overline F\setminus\{e_i\}$ is a feasible solution for the DST instance}
\STATE $\overline F\gets \overline F\setminus\{e_i\}$.
\ENDIF
\ENDFOR
\RETURN $\overline F$
\end{algorithmic}
\label{our alg}
\end{algorithm}

\section{The analysis}\label{the analysis}
The general framework for analyzing primal-dual algorithms is to use the dual constraints to relate the cost of purchased edges and the dual variables. However, here we do not use the dual constraints and rather we use the buckets we created for each edge. Recall $\overline F$ is the solution output by Algorithm \ref{our alg}. We define $\overline{F}_{\killer}$ to be the set of edges in $\overline F$ that was purchased as killer edge\footnote{See Definition \ref{def killer and expansion edge}.}. Similarly define $\overline{F}_{\expansion}$ and $\overline F_{\antenna}$. For each iteration $l$, we denote by $\fs$ the set $F$ at this iteration, $\mathcal{A}_l$ denotes the set of active moats with respect to $\fs$, and $\epsilon_l$ is the amount we increased the dual variables (corresponding to active moats) with at iteration $l$. Finally, Let $y^*$ be the dual solution for \eqref{dual const modified} constructed in the course of the algorithm. We use the following notation throughout the analysis.

\begin{definition}
Fix an iteration $l$. For any $A \in \mathcal{A}_l$, let
\[
\Delta^{l}_{\killer}(A):=\{e\in\overline F_{\killer}:~e~is~killer~with~respect~to~A~under~\fs\},
\]
in other words, $\Delta^{l}_{\killer}(A)$ is the set of all killer edges in $\overline F$ such that they are killer edge with respect to $A$ at iteration $l$. Similarly define $\Delta^{l}_{\expansion}(A)$.

Let $\Delta^{l}_{\antenna}(A):=\{e\in\overline F_{\antenna}:~e\in\delta^{in}(A)\}$. Finally, we define
\[
\Delta^{l}(A):=\Delta^{l}_{\killer}(A)\cup\Delta^{l}_{\expansion}(A)\cup\Delta^{l}_{\antenna}(A).
\]
Note $\Delta^{l}_{\killer}(A)$, $\Delta^{l}_{\expansion}(A)$, and $\Delta^{l}_{\antenna}(A)$ are pairwise disjoint for any $A\in\mathcal{A}_l$.
\end{definition}

Suppose we want to show that the performance guarantee of Algorithm \ref{our alg} is $2\cdot\alpha$ for some $\alpha\geq 1$, it suffices to show the following: for any iteration $l$ we have

\begin{equation}\label{ineqaulity primal-dual analysis}
    \sum\limits_{S\in\mathcal{A}_l}|\Delta^{l}(S)|\leq\alpha\cdot|\mathcal{A}_l|.
\end{equation}

Once we have (\ref{ineqaulity primal-dual analysis}), then the $(2\cdot\alpha)$-approximation follows easily:
\begin{align}
        \sum\limits_{e\in\overline F}c(e)&= \sum\limits_{e\in\overline F_{\killer}}\sum\limits_{l}\sum\limits_{\substack{S\in\mathcal{A}_l:\\e\in\Delta^{l}_{\killer}(S)}}\epsilon_l+\sum\limits_{e\in\overline F_{\expansion}}\sum\limits_{l}\sum\limits_{\substack{S\in\mathcal{A}_l:\\e\in\Delta^{l}_{\expansion}(S)}}\epsilon_l+\sum\limits_{e\in\overline F_{\antenna}}\sum\limits_{l}\sum\limits_{\substack{S\in\mathcal{A}_l:\\e\in\Delta^{l}_{\antenna}(S)}}\epsilon_l\\
        &=\sum\limits_{l}\epsilon_l\cdot\sum\limits_{S\in\mathcal{A}_l}|\Delta^{l}(S)|\label{5.4}\\
        &\leq\alpha\cdot\sum\limits_{l}|\mathcal{A}_l|\epsilon_l \\
        &=\alpha\cdot\sum\limits_{S\subseteq V\setminus\{r\}}y^*_S \label{5.5}\\
        &\leq \alpha\cdot\big(2\cdot\opt(\ref{dual lp})\big) \label{5.6}\\
        &=2\cdot\alpha\cdot\opt(\ref{lp:primal}) \label{5.7}\\
        &\leq 2\cdot\alpha\cdot\opt \label{4.7},
\end{align}
where the first equality follows from the algorithm, the second equality is just an algebraic manipulation, (\ref{5.4}) follows from \eqref{ineqaulity primal-dual analysis}. Equality (\ref{5.5}) follows from the fact we uniformly increased the dual variables corresponding to active moats by $\epsilon_l$ at iteration $l$, (\ref{5.6}) follows from feasibility of $\frac{y^*}{2}$ for (\ref{dual lp}), and (\ref{5.7}) follows from strong duality theorem for linear programming. 

It remains to show (\ref{ineqaulity primal-dual analysis}) holds. Consider iteration $l$. Using the bound on the total degree of nodes in $G$ (using minor-free properties) to show (\ref{ineqaulity primal-dual analysis}), it suffices to bound the number of edges in $\bar F_{\antenna}\cup\bar F_{\killer}\cup \bar F_{\expansion}$ that are being paid by some active moat at iteration $l$, by $O(|\mathcal{A}_l|)$. We provide charging schemes for each type of edges, separately. Since $G$ is quasi-bipartite, it is easy to show that for each active moat $A\in\mathcal{A}_l$, there is at most one antenna edge in $\bar F$ that enters $A$, this is proved in Section \ref{section bound on antenna edges}. The charging scheme for killer edges is also simple as one can charge a killer edge to an active moat that it kills; this will be formalized in Section \ref{section bound on killer edges}. However, the charging scheme for expansion edges requires more care and novelty. The difficulty comes from the case that an expansion edge is not pruned because it would disconnect some terminals that are not part of any active moat that $e$ is entering this iteration.

Our charging scheme for expansion edges is more global. In a two-stage process, we construct an auxiliary tree that encodes some information about which nodes can be reached from SCCs using edges in $F_l$ (which is the information we used in the definition of expansion edge). Then using a token argument, we leverage properties of our construction to show the number of expansion edges is at most twice the number of active moats in any iteration. These details are presented in \ref{section bound on expansion edges}. Finally, in Section \ref{putting everything together} we put all the bounds we obtained together and derive our approximation factors.

\subsection{Counting the number of antenna edges in an iteration}\label{section bound on antenna edges}
Fix an iteration $l$. Recall $\fs$ denotes the set $F$ at iteration $l$, and $\mathcal{A}_l$ denotes the set of active moats with respect to $\fs$. It is easy to bound the number of antenna edges in $\overline F$ against $|\mathcal{A}_l|$. We do this in the next lemma.

\begin{lemma}\label{number of antenna edges}
At the beginning of each iteration $l$, we have $\sum\limits_{A\in\mathcal{A}_l}|\Delta^{l}_{\antenna}(A)|\leq |\mathcal{A}_l|$.
\end{lemma}
\begin{proof}
Suppose an active moat $A\in\mathcal{A}_l$ is paying toward at least two antenna edges $e=(u,v)$ and $f=(u',v')$ that are in $\overline F$. Let $C_A$ be the SCC part of $A$. Note that since $e$ and $f$ are antenna edges, $u$ and $u'$ are Steiner nodes. Together with the fact that the graph is quasi-bipartite, the heads $v$ and $v'$ are terminals and therefore contained in $C_A$. Since all the edges in $C_A$ are bought before $e$ and $f$, one of $e$ or $f$ should have been pruned in the deletion phase, a contradiction. Hence, $|\Delta^{l}_{\antenna}(A)|\leq 1$ which implies the desired bound.
\end{proof}

\subsection{Counting the number of killer edges in an iteration}\label{section bound on killer edges}
We introduce a notion called {\em alive} terminal which helps us to bound the number of killer edges at a fixed iteration against the number of active moats in that iteration. Also this notion explains the name killer edge. Throughout the algorithm, we show every active moat contains exactly one alive terminal and every alive terminal is in an active moat. 

We consider how terminals can be ``killed'' in the algorithm by associating active moats with terminals that have not yet been part of a moat that was killed.
At the beginning of the algorithm, we mark every terminal alive, note that every singleton terminal set is initially an active moat as well. Let $e_l=(u,v)$ be the edge that was added to $\fs$ at iteration $l$. If $e_l=(u,v)$ is a non-antenna edge, then for every active set $A$ such that $e_l$ is a killer edge with respect to $A$ under $\fs$, mark the alive terminal in $A$ as {\em dead}\footnote{It is possible, $e_l$ is bought as an expansion edge but kills some alive terminals. For example, in Figure \ref{pic for relation between edges and moats} suppose $e$ is being added to $\fs$ at iteration $l$ as an expansion edge (note that $A$ pays toward the expansion bucket of $e$). Then, we mark the alive terminal in $A'$ as dead because $e$ is a killer edge with respect to $A'$ under $\fs$.}. If $e_l=(u,v)$ is an antenna edge, then for every active moat $A$ such that $e_l\in\delta^{in}(A)$ and $C_A$ is {\bf not} in any active moat with respect to $\fs\cup\{e_l\}$, then mark the alive terminal in $A$ as dead\footnote{For example, suppose the antenna edge $e_l=(u,v)\in\delta^{in}(A)$ is being added to $\fs$ and $u$ is in $C_{A'}$ for some active moat $A'$. Then, after adding $e_l$ to $\fs$, we mark the alive terminal in $A$ as dead.}.

The important observation here is that by definition, if $e_l$ is a killer edge, then there must be an active set that satisfies the above condition, hence there is at least one alive terminal that will be marked dead because of $e_l$. In the case that $e_l$ is bought as killer edge, arbitrarily pick an alive terminal $t_{e_l}$ that dies because of $e_l$ and assign $e_l$ to $t_{e_l}$. Note that $t_{e_l}$ was alive until $e_l$ was added to $\fs$. 

\begin{definition}\label{important killer edges}
Fix an iteration $l$. We define
\[
\overline F^l_{\killer}:=\{e\in\overline F_{\killer}:~\exists A\in\mathcal{A}_l~s.t.~e\in\Delta^{l}_{\killer}(A)\},
\]
in other words, $\overline F^l_{\killer}$ is the set of all killer edges in $\overline F$ such that some active moat(s) is paying toward their killer bucket at iteration $l$.
\end{definition}

Now we can state the main lemma of this section.
\begin{lemma}\label{bound for the number of killer edges}
At the beginning of each iteration $l$, we have $|F^l_{\killer}|\leq |\mathcal{A}_l|$.
\end{lemma}
\begin{proof}
As shown above, every killer edge $e$ is assigned to a terminal $t_e$ that was alive until $e$ was added to $F$. Thus, at iteration $l$ all the edges in $\overline F_{\killer}\setminus\fs$ correspond to a terminal that is alive at this iteration. Since there is a one-to-one correspondence between alive terminals and active sets, the number of edges in $\overline F_{\killer}\setminus \fs$ is at most $|\mathcal{A}_l|$. The lemma follows by noticing that $\overline F^l_{\killer}\subseteq \overline F_{\killer}\setminus\fs$.
\end{proof}

Note that the above lemma does not readily bound $\sum\limits_{A\in\mathcal{A}_l}|\Delta^l_{\killer}(A)|$ against $|\mathcal{A}_l|$ which is required to prove inequality (\ref{ineqaulity primal-dual analysis}). We need the properties of minor-free graphs to do so. In the next section we prove a similar result for expansion edges and then using the properties of the underlying graph, we demonstrate our approximation guarantee

\subsection{Counting the number of expansion edges in an iteration}\label{section bound on expansion edges}
%{\bf ZF: High level description of the approach}
The high level idea to bound the number of expansion edges is to look at the graph $\overline F\cup\fs$ and contract all SCCs\footnote{Recall that we do NOT call a Steiner node that is a singleton strongly connected component of $(V,\fs)$ an SCC. So every SCC in $(V,\fs)$ is either $\{r\}$ or contains at least one terminal node.} of $(V,\fs)$. Then, we construct an auxiliary tree that highlights the role of expansion edges to the connectivity of active moats. Then, using this tree we provide our charging scheme and show the number of edges in $\overline F_{\expansion}$ that are being paid by some active moats at iteration $l$ is at most twice the number of active moats.

%One challenge is that an expansion edge $e$ is not pruned because it would disconnect some terminals that are not part of any active set that $e$ is entering this iteration $l$. So our counting argument must be more global. In a two-stage process, we construct an auxiliary tree that encodes some information about which nodes can be reached from SCCs using edges in $F_l$ (which is the information we used in the definition of expansion edge). Then using a token argument, we leverage properties of our construction to show the number of expansion edges is at most twice the number of active moats in any iteration. 

%Then, we consider a shortest path tree rooted at $r$ where a shortest path from $u$ to $v$ is a $u-v$-path with minimum number of expansion edges. Finally, we modify this arborescence to obtain another arborescence rooted at $r$ that has a \say{nice} structure which allows us to present a token argument that shows the number of expansion edges is at most twice the number of active moats at any iteration.

We fix an iteration $l$ for this section. First let us recall some notation and definition that we use extensively in this section.

\begin{itemize}
    \item $\overline F$ is the output solution of the algorithm.
    \item $\fs\subseteq E$ is the set of purchased edges in the growing phase up to the beginning of iteration $l$ (i.e., set $F$ in the algorithm at iteration $l$).
    \item $\mathcal{A}_l$ is the set of active moats with respect to $\fs$ (see Definition \ref{active moats}). Recall each $A\in\mathcal{A}_l$ is consist of an SCC (with respect to edges in $\fs$) and bunch of Steiner nodes. Denote by $C_A$ the SCC part of $A$.
\end{itemize}

We define an analogue of Definition \ref{important killer edges} for expansion edges.
\begin{definition}\label{expansion edge that matter}
Fix an iteration $l$. Then, we define
\[
\overline F^l_{\expansion}:=\{e\in\overline F_{\expansion}:~\exists A\in\mathcal{A}_l~s.t.~e\in\Delta^{l}_{\expansion}(A)\},
\]
in other words, $\overline F^l_{\expansion}$ is the set of all expansion edges in $\overline F\setminus\fs$ such that some active moat(s) is paying toward their expansion bucket at iteration $l$.
\end{definition}

This section is devoted to prove the following inequality.
\begin{lemma}\label{the bound on the number of expansion edges}
At the beginning of each iteration $l$ of the algorithm, we have $|\overline F^l_{\expansion}|\leq 2\cdot|\mathcal{A}_l|$.
\end{lemma}

\noindent{\bf Sketch of the proof:}
We start by giving a sketch of the proof of Lemma \ref{the bound on the number of expansion edges}. Consider the subgraph $\fs\cup \overline F$ of $G$. Contract every SCC of $(V,\fs)$ and denote the resulting subgraph by $H$ (keeping all copies of parallel edges that may result). For every non-root, non-Steiner node $v\in V(H)$, we call $v$ active if it is a contraction of an SCC that is a subset of an active moat in $\mathcal{A}_l$, otherwise we call $v$ inactive. Note that $r$ is a singleton SCC in $(V,\fs)$ and therefore $r\in V(H)$. We call an edge in $E(H)$ an expansion edge, if its corresponding edge is in $\overline F^l_{\expansion}$. Note that every non root vertex in $V(H)$ is either labeled active/inactive, or it is a Steiner node. Lemma \ref{the bound on the number of expansion edges} follows if we show the number of expansion edges in $H$ is at most twice the number of active vertices in $H$. As we stated at the beginning of this section, we use an arborescence that highlights the role of expansion edges to the connectivity of active vertices in $H$. A bit more formally, we show if every expansion edge is \say{good} with respect to the arborescence, which is formalized below, then every expansion edge is \say{close} to an active vertex in $H$ and we use this in our charging scheme.

Given an arborescence $T$, define $\explevel_T(v)$ to be the expansion level of $v$ with respect to $T$, i.e., the number of expansion edges on the dipath from $r$ to $v$ in $T$.

\begin{definition}\label{def, good/bad exp}
Given an arborescence $T$ and an expansion edge $e=(u,v)$, we say $e$ is a {\em good} expansion edge with respect to $T$ if one of the following cases happens:
\begin{itemize}
    \item Type 1: If $u$ has an active ancestor $w$ such that $\explevel_T(w)=\explevel_T(u)$.
    \item Type 2: If $e$ is not of type 1 and the subtree rooted at $u$ has an active vertex $w$ such that $\explevel_T(w)\leq \explevel_T(u)+1$.
\end{itemize}
Every expansion edge that is not of type 1 or type 2, is called a {\em bad} expansion edge with respect to $T$.
\end{definition}

A starting point for an arborescence that every expansions edge is good, is a shortest path arborescence rooted at $r$ in $H$ where each expansion edge has cost $1$ and the rest of the edges have cost $0$. However, as Figure \ref{pic highlevel idea expansion} shows, there could be some {\em bad} expansion edges in this arborescence. For example, $e$ is a {\em bad} expansion edge with respect to the arborescence in Figure \ref{pic highlevel idea expansion} (b). Since $B_2$, the tail of $e$, is an inactive vertex, there must be an active vertex, namely $A_3$, that has a dipath from $A_3$ to $B_2$ in $F_\ell$ (see Claim \ref{properties of the endpoints of expansion edge}). Then, we can \say{cut} the subtree rooted at $B_2$ and \say{paste} it under $A_3$ as shown in Figure \ref{pic highlevel idea expansion}(c). It is easy to verify that now every expansion edge is good with respect to the arborescence in Figure \ref{pic highlevel idea expansion}(c). We formalize this \say{cut and paste} procedures in Algorithm \ref{modifying bfs alg} and prove the output of the algorithm is an arborescence with the property that every expansion edge is good. At the end of this section, given an arborescence that every expansion edge is good, we show there is a rather natural charging scheme that proves Lemma \ref{the bound on the number of expansion edges}.

\begin{figure}[t]
\centering
\begin{minipage}{0.3\textwidth}
\centering
\begin{tikzpicture}[smallnode/.style={circle, draw=black, fill=black, very thick, scale =0.6},
bignode/.style={circle, draw=black, very thick, scale =.5},scale=.4]

\draw[blue, dashed] (-3, -2) ellipse (.7 and 2);

\draw[blue, dashed] (-2, -8) ellipse (2 and .7);

\draw[blue, dashed] (3.5, -1) ellipse (.7 and 2.2);

\draw[blue, dashed, rotate around={-15:(0.2,-3.7)}] (0.2, -3.7) ellipse (.6 and 1.7);

\draw[blue, dashed, rotate around={30:(1,-3.7)}] (1, -3.7) ellipse (.8 and 1.9);

\draw[blue, dashed] (2, -2.8) ellipse (2.4 and .7);

\draw[blue, dashed] (1.5, -6) ellipse (.7 and 1.7);

\node[smallnode, red, label=above:$r$] (r) at (0,0) {};
\node[bignode, blue] (A1) at (-3,-1){$A_1$};
\node[bignode, blue] (A2) at (3.5,0){$A_2$};
\node[bignode, blue] (A3) at (0.5,-3){$A_3$};
\node[bignode, blue] (A4) at (1.5,-7){$A_4$};
\node[bignode, blue] (A5) at (-3,-8){$A_5$};

\node[bignode, red] (B1) at (-1.5,-2.5){$B_1$};
\node[bignode, red] (B2) at (-3,-5){$B_2$};
\node[bignode, red] (B3) at (-1.5,-6){$B_3$};
\node[bignode, red] (B4) at (1,-1){$B_4$};
\node[bignode, red] (B5) at (3.5,-4.5){$B_5$};

\node[smallnode] (s1) at (-3,-3){};
\node[smallnode] (s2) at (0,-5){};
\node[smallnode] (s3) at (-1,-8){};

\node[smallnode] (s4) at (3.5,-2.7){};
\node[smallnode] (s5) at (1.5,-5){};

\draw[->, blue, very thick] (r) -- (A1) {};
\draw[->, very thick] (A1) -- (B1) {};
\draw[->, green, very thick] (B1) -- (s1) {};
\draw[->, very thick] (s1) -- (A1) {};

\draw[->, very thick] (s4) -- (A3) {};
\draw[->, very thick] (A3) -- (B4) {};
\draw[->, very thick] (A3) -- (B5) {};
\draw[->, very thick] (A3) -- (B2) {};
\draw[->, very thick] (s1) -- (B2) {};
\draw[->, green, very thick] (B2) -- (s2) node[black, midway, above] {$e$};
\draw[->, very thick] (s2) -- (A3) {};

\draw[->, very thick] (A5) -- (B3) {};
\draw[->, very thick] (s2) -- (B3) {};
\draw[->, green, very thick] (B3) -- (s3) {};
\draw[->, very thick] (s3) -- (A5) {};

\draw[->, blue, very thick] (r) -- (A2) {};
\draw[->, very thick] (A2) -- (B4) {};
\draw[->, green, very thick] (B4) -- (s4) {};
\draw[->, very thick] (s4) -- (A2) {};

\draw[->, very thick] (A4) -- (B5) {};
\draw[->, very thick] (s4) -- (B5) {};
\draw[->, green, very thick] (B5) -- (s5) {};
\draw[->, very thick] (s5) -- (A4) {};
\draw[->, very thick] (s5) -- (A3) {};
\end{tikzpicture}
\subcaption{}
\end{minipage}
\begin{minipage}{0.3\textwidth}
\centering
\begin{tikzpicture}[smallnode/.style={circle, draw=black, fill=black, very thick, scale =0.4},
bignode/.style={circle, draw=black, very thick, scale =1.5},scale=.4]

\node[smallnode, red, label=above:$r$] (r) at (0,0) {};
\node[smallnode, blue] (a1) at (-1,-1){};
\node[smallnode, blue] (a2) at (1,-1){};
\node[smallnode, red] (b1) at (-1,-2){};
\node[smallnode] (s1) at (-1,-3){};
\node[smallnode, red, label=left:$B_2$] (b2) at (-1,-4){};
\node[smallnode] (s2) at (-1,-6){};
\node[smallnode, red] (b3) at (-1,-7){};
\node[smallnode] (s3) at (-1,-8){};
\node[smallnode, blue] (a5) at (-1,-9){};

\draw[->, very thick, blue] (r) -- (a1) {};
\draw[->, very thick] (a1) -- (b1) {};
\draw[->, very thick, green] (b1) -- (s1) {};
\draw[->, very thick] (s1) -- (b2) {};
\draw[->, very thick, green] (b2) -- (s2) node[midway, left, black] {$e$};
\draw[->, very thick] (s2) -- (b3) {};
\draw[->, very thick, green] (b3) -- (s3) {};
\draw[->, very thick] (s3) -- (a5) {};

\node[smallnode, blue] (a2) at (1,-1){};
\node[smallnode, red] (b4) at (1,-2){};
\node[smallnode] (s4) at (1,-3){};
\node[smallnode, blue,label=right:$A_3$] (a3) at (1,-4){};
\node[smallnode, red] (b5) at (1,-5){};
\node[smallnode] (s5) at (1,-6){};
\node[smallnode, blue] (a4) at (1,-7){};

\draw[->, very thick, blue] (r) -- (a2) {};
\draw[->, very thick] (a2) -- (b4) {};
\draw[->, very thick, green] (b4) -- (s4) {};
\draw[->, very thick] (s4) -- (a3) {};
\draw[->, very thick] (a3) -- (b5) {};
\draw[->, very thick, green] (b5) -- (s5) {};
\draw[->, very thick] (s5) -- (a4) {};
\end{tikzpicture}
\subcaption{}
\end{minipage}
\begin{minipage}{0.3\textwidth}
\centering
\begin{tikzpicture}[smallnode/.style={circle, draw=black, fill=black, very thick, scale =0.4},
bignode/.style={circle, draw=black, very thick, scale =1.5},scale=.4]

\node[smallnode, red, label=above:$r$] (r) at (0,0) {};
\node[smallnode, blue] (a1) at (-1,-1){};
\node[smallnode, blue] (a2) at (1,-1){};
\node[smallnode, red] (b1) at (-1,-2){};
\node[smallnode] (s1) at (-1,-3){};
\node[smallnode, red, label=left:$B_2$] (b2) at (-1,-4){};
\node[smallnode] (s2) at (-1,-6){};
\node[smallnode, red] (b3) at (-1,-7){};
\node[smallnode] (s3) at (-1,-8){};
\node[smallnode, blue] (a5) at (-1,-9){};

\draw[->, very thick, blue] (r) -- (a1) {};
\draw[->, very thick] (a1) -- (b1) {};
\draw[->, very thick, green] (b1) -- (s1) {};
%\draw[->, very thick] (s1) -- (b2) {};
\draw[->, very thick, green] (b2) -- (s2) node[midway, left, black] {$e$};
\draw[->, very thick] (s2) -- (b3) {};
\draw[->, very thick, green] (b3) -- (s3) {};
\draw[->, very thick] (s3) -- (a5) {};

\node[smallnode, blue] (a2) at (1,-1){};
\node[smallnode, red] (b4) at (1,-2){};
\node[smallnode] (s4) at (1,-3){};
\node[smallnode, blue,label=right:$A_3$] (a3) at (1,-4){};
\node[smallnode, red] (b5) at (1,-5){};
\node[smallnode] (s5) at (1,-6){};
\node[smallnode, blue] (a4) at (1,-7){};

\draw[->, very thick] (a3) -- (b2) {};
\draw[->, very thick, blue] (r) -- (a2) {};
\draw[->, very thick] (a2) -- (b4) {};
\draw[->, very thick, green] (b4) -- (s4) {};
\draw[->, very thick] (s4) -- (a3) {};
\draw[->, very thick] (a3) -- (b5) {};
\draw[->, very thick, green] (b5) -- (s5) {};
\draw[->, very thick] (s5) -- (a4) {};
\end{tikzpicture}
\subcaption{}
\end{minipage}
\caption{(a) shows part of the subgraph $\fs\cup\bar F$ of $G$, in particular, the SCCs of $(V,\fs)$ are shown with circles but the nodes inside SCCs are not shown for simplicity. The blue SCCs are inside some active moats shown with dashed ellipses. Contracting all the SCCs result in the graph $H$ discussed before. Black edges are in $\fs$, blue edges are in $\bar F\setminus\fs$, and green edges are in $\bar F^l_{\expansion}$. In (b), we have a shortest path arborescence rooted at $r$ where the cost of edges is one if it is green and zero otherwise. Note that $e$ is a bad expansion edge with respect to this arborescence. In (c), we show how to construct an arborescence using cut-and-paste procedure so that every expansion edge is a good expansion edge in the resulting arborescence.}
\label{pic highlevel idea expansion}
\end{figure}
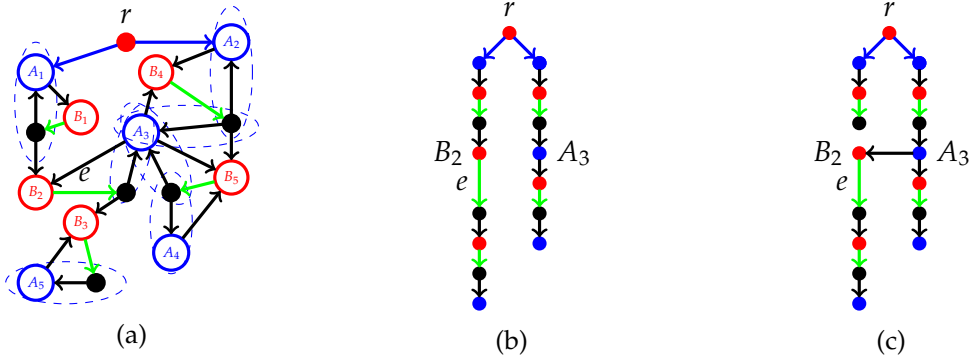

\noindent{\bf Detailed proof:} Our arguments use the following observations about edges being paid as expansion edges in this iteration.
\begin{claim}\label{properties of the endpoints of expansion edge}
Let $e=(u,v)\in\overline F^l_{\expansion}$, then $u\in X$, $v\in A$ for some $A\in\mathcal{A}_l$, and there is a $\fs$-path from $C_A$ to $u$. Furthermore, the SCC of $(V,\fs)$ that contains $u$ is not contained in any active moat in $\mathcal{A}_l$.
\end{claim}
\begin{proof}
Since $e$ is an expansion edge with respect to $A$ by Definition \ref{relation between edges and active moats} there exists $A'\subsetneq V$ that is active with respect to $\fs\cup\{e\}$. Since $e$ is a non-antenna edge, $u$ must be a terminal. Furthermore, $u\neq r$ because $A'$ is active so $u\in X$. By Lemma \ref{remark on active moats}, the SCC part $C_{A'}$ of $A'$ contains both $u$ and all vertices in $C_A$, hence there is a dipath in $\fs\cup\{e\}$ from $C_A$ to $u$. However, notice that this dipath cannot contain $e$, thus the path is actually a $\fs$-path. Finally, since there is a $\fs$-path from $C_A$ to $u$, the SCC $B$ of $\fs$ that contains $u$ is not a violated set and therefore no active moat in $\mathcal{A}_l$ contains $B$.
\end{proof}

Recall the definition of graph $H$. We state couple of facts about this graph which will be useful later.

\begin{claim}\label{inactive vertices are reachable}
For every inactive vertex $v$ in $H$, there is a $\fs$-path from either $r$ or an active vertex to $v$.
\end{claim}
\begin{proof}
Let $v$ be the contraction of SCC $B$. Consider all SCCs in $(V,\fs)$ that $B$ is reachable from via a $\fs$-path and pick such SCC $C$ that is not reachable from any other SCCs of $(V,\fs)$, it is easy to see that either $C=\{r\}$ or $C$ is inside an active moat and therefore, $v$ is reachable from the active vertex that is the contraction of $C$.
\end{proof}

\begin{claim}\label{expansion endpoints}
For every expansion edge $e=(u,v)\in E(H)$, $u$ must be inactive and $v$ is either active or a Steiner node.
\end{claim}
\begin{proof}
Let $e'=(u',v')\in\overline F^l_{\expansion}$ be the corresponding edge to $e$. By Claim \ref{properties of the endpoints of expansion edge}, $u'\in X$ and the SCC $B$ in $(V,\fs)$ that contains $u'$ is not a subset of any active moat in $\mathcal{A}_l$. Therefore, $u$ is the contraction of such SCC $B$ and so it is labeled inactive. Again by Claim \ref{properties of the endpoints of expansion edge}, $v'\in A$ for some $A\in\mathcal{A}_l$. If $v'$ is not a Steiner node (and therefore $v$ is not a Steiner node) then $v'\in C_A$ and $v$ is the contraction of $C_A$ and so it is labeled active.
\end{proof}

To simplify the exposition, we use the following auxiliary graph instead of $H$ in proving the main lemma of this section. With abuse of notation, we say a dipath in $H$ is a $\fs$-path if its corresponding dipath in $\fs\cup\overline F$ is a $\fs$-path.
\begin{definition}[Auxiliary graph $\auxGraph$]\label{auxiliary graph}
For every expansion edge $e=(u,v)$ in $H$ and every active vertex $w$ in $H$ such that there is a $\fs$-path from $w$ to $u$, add an auxiliary edge $(w,u)$\footnote{We might create parallel edges but since at the end we work with arborescence, the parallel edges do not matter.}. Set the cost of each expansion edge to $1$ and the rest of the edges (including the auxiliary edges) have cost $0$. Denote this graph by $\auxGraph$.
\end{definition}

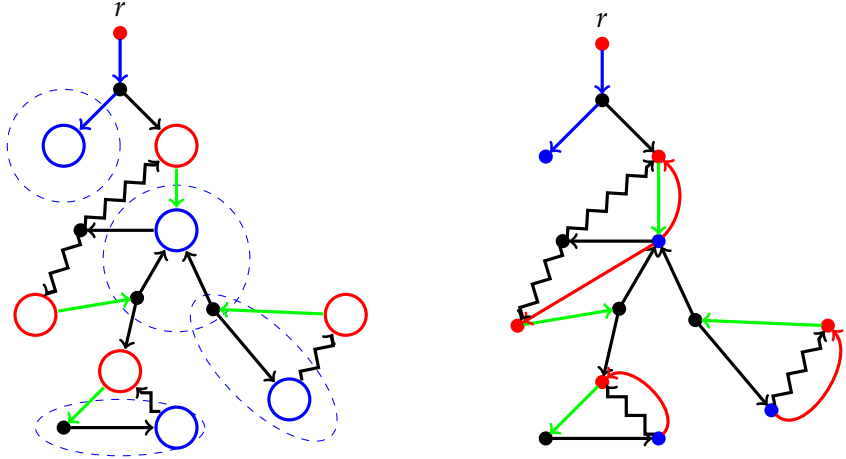
\begin{figure}
\centering
\begin{minipage}{0.4\textwidth}
\begin{tikzpicture}[smallnode/.style={circle, draw=black, fill=black, very thick, scale =0.4},
bignode/.style={circle, draw=black, very thick, scale =1.5},scale=.75]

\node[smallnode, red, label=above:$r$] (r) at (0,0) {};
\node[smallnode] (s1) at (0,-1){};
\node[bignode, red] (b1) at (1,-2){};
\node[bignode, blue] (a1) at (-1,-2){};
\node[bignode, blue] (a2) at (1,-3.5){};
\node[smallnode] (s2) at (0.3,-4.7){};
\node[smallnode] (s3) at (1.65,-4.9){};
\node[bignode, red] (b2) at (-1.5,-5){};
\node[bignode, red] (b3) at (0,-6){};
\node[smallnode] (s4) at (-1,-7){};
\node[bignode, blue] (a3) at (1,-7){};
\node[bignode, red] (b4) at (4,-5){};
\node[bignode, blue] (a4) at (3,-6.5){};
\node[smallnode] (h) at (-.7,-3.5){};

\draw[blue, dashed] (1, -4) circle (1.3);
\draw[blue, dashed] (-1, -2) circle (1);
\draw[blue, dashed, rotate=45] (-2.4, -6) ellipse (.7 and 1.7);
\draw[blue, dashed] (0, -7) ellipse (1.5 and .5);

\draw[->, blue, very thick] (r) -- (s1) {};
\draw[->, blue, very thick] (s1) -- (a1) {};
\draw[->, very thick] (s1) -- (b1) {};
\draw[->, green, very thick] (b1) -- (a2) {};
\draw[->, very thick] (s2) -- (a2) {};
\draw[->, very thick] (s3) -- (a2) {};
\draw[->, green, very thick] (b2) -- (s2) {};
\draw[->, green, very thick] (b4) -- (s3) {};
\draw[->, very thick] (s3) -- (a4) {};
\draw[->, very thick] (s2) -- (b3) {};
\draw[->, green, very thick] (b3) -- (s4) {};
\draw[->, very thick] (s4) -- (a3) {};
\draw[->, very thick] (a2) -- (h) {};

\draw[->, very thick, decorate, decoration=zigzag] (h) -- (b1);
\draw[->, very thick, decorate, decoration=zigzag] (h) -- (b2);
\draw[->, very thick, decorate, decoration=zigzag] (a4) -- (b4);
\draw[->, very thick, decorate, decoration=zigzag] (a3) -- (b3);
\end{tikzpicture}
\end{minipage}
\begin{minipage}{0.4\textwidth}
\begin{tikzpicture}[smallnode/.style={circle, draw=black, fill=black, very thick, scale =0.4},
bignode/.style={circle, draw=black, very thick, scale =1.5},scale=.75]

\node[smallnode, red, label=above:$r$] (r) at (0,0) {};
\node[smallnode] (s1) at (0,-1){};
\node[smallnode, red] (b1) at (1,-2){};
\node[smallnode, blue] (a1) at (-1,-2){};
\node[smallnode, blue] (a2) at (1,-3.5){};
\node[smallnode] (s2) at (0.3,-4.7){};
\node[smallnode] (s3) at (1.65,-4.9){};
\node[smallnode, red] (b2) at (-1.5,-5){};
\node[smallnode, red] (b3) at (0,-6){};
\node[smallnode] (s4) at (-1,-7){};
\node[smallnode, blue] (a3) at (1,-7){};
\node[smallnode, red] (b4) at (4,-5){};
\node[smallnode, blue] (a4) at (3,-6.5){};
\node[smallnode] (h) at (-.7,-3.5){};

\draw[->, blue, very thick] (r) -- (s1) {};
\draw[->, blue, very thick] (s1) -- (a1) {};
\draw[->, very thick] (s1) -- (b1) {};
\draw[->, green, very thick] (b1) -- (a2) {};
\draw[->, very thick] (s2) -- (a2) {};
\draw[->, very thick] (s3) -- (a2) {};
\draw[->, green, very thick] (b2) -- (s2) {};
\draw[->, green, very thick] (b4) -- (s3) {};
\draw[->, very thick] (s3) -- (a4) {};
\draw[->, very thick] (s2) -- (b3) {};
\draw[->, green, very thick] (b3) -- (s4) {};
\draw[->, very thick] (s4) -- (a3) {};
\draw[->, very thick] (a2) -- (h) {};

\draw[->, red, very thick] (a2) to [out=45,in=-40] (b1);
\draw[->, red, very thick] (a2) -- (b2);
\draw[->, red, very thick] (a4) to [out=-45,in=-45]  (b4);
\draw[->, red, very thick] (a3) to [out=45,in=40]  (b3);

\draw[->, very thick, decorate, decoration=zigzag] (h) -- (b1);
\draw[->, very thick, decorate, decoration=zigzag] (h) -- (b2);
\draw[->, very thick, decorate, decoration=zigzag] (a4) -- (b4);
\draw[->, very thick, decorate, decoration=zigzag] (a3) -- (b3);
\end{tikzpicture}
\end{minipage}
\caption{The left picture shows the subgraph $\fs\cup\overline F$. The SCCs of $(V,\fs)$ is shown with circles and the blue ones are inside some active moats shown with dashed ellipses at iteration $l$. Zigzag paths and black edges are in $\fs$, blue edges are in $\overline F\setminus\fs$, and green edges are in $\overline F^l_{\expansion}$. The right picture shows $\auxGraph$ constructed from $\fs\cup\overline F$. The red edges are the auxiliary edges.}
\label{pic for auxiliary graph}
\end{figure}
%{\bf ZF: The only thing this might be a bit misleading on is that the black paths might, in general, be more like mini arborescences.}

See Figure \ref{pic for auxiliary graph} for an illustration of $\auxGraph$. Given a subset $T\subseteq E(\auxGraph)$, we say $e\in\overline F^l_{\expansion}$ is in $T$ if its corresponding expansion edge in $E(\auxGraph)$ is in $T$. For the rest of this section, when we talk about arborescence we mean an arborescence rooted at $r$ that is a subgraph of $\auxGraph$ and every active/inactive vertices in $V(\auxGraph)$ is reachable from $r$ in this arborescence. Following are two properties of arborescences that will be useful.
\begin{lemma}\label{expansion edges present in arborescence}
Let $T$ be an arborescence rooted at $r$ in $\auxGraph$. Then, we have
\begin{itemize}
    \item[a.] Every edge in $\overline F^l_{\expansion}$ is in $T$ as well. And
    \item[b.] For every expansion edge $e=(u,v)$ in $T$, either $v$ is active or the subtree $T_v$ of $T$ rooted at $v$ contains an active vertex.
\end{itemize}
\end{lemma}
\begin{proof}
Proof of part (a): note that all edges in $\overline F^l_{\expansion}$ are present in $\auxGraph$. Suppose $e\in\overline F^l_{\expansion}$ that is not in $T$. Replace every auxiliary edge with its corresponding $\fs$-path in $T$. Note that the resulting subgraph $H'$ is a subgraph of $H$ (recall $H$ is the contracted graph obtained from $\fs\cup \overline F$) and every active/inactive vertex is still reachable from $r$ in $H'$. Replace every active/inactive vertex in $H'$ by its corresponding contracted SCC, this is a subgraph of $(\overline F\cup\fs)\setminus\{e\}$ and every terminal is reachable from $r$. Therefore, $(\fs\cup \overline F)\setminus\{e\}$ is a feasible solution for the DST instance. Since $e$ was added to $F$ after all edges in $\fs$, in the deletion phase we should have pruned $e$, a contradiction with the fact that $e\in\overline F$.

Proof of part (b): suppose not. Then $v$ is a Steiner node and every vertex in the subtree rooted at $v$ is either inactive or a Steiner node. If it is inactive then by Claim \ref{inactive vertices are reachable} there must be a $\fs$-path from either an active vertex or $r$ to it. Add these $\fs$-path for all inactive vertices in $T_v$. With the same argument as in part (a) we conclude that $(u,v)$ (i.e., its corresponding edge in $\overline F^l_{\expansion}$) should have been pruned, a contradiction.
\end{proof}

Denote by $\overline T$ the shortest path tree in $\auxGraph$ rooted at $r$. In the following we show how to turn $\overline T$ to an arborescence such that every expansion edge is a good expansion edge (with respect to the resulting arborescence). Once we have that, we can provide a charging argument that proves the main lemma of this section (i.e., Lemma \ref{the bound on the number of expansion edges}). We use the following algorithm for modifying $\overline T$, note that this is for the analysis and our primal-dual algorithm does not use this.

\begin{algorithm}[H]
\caption{Modifying $\overline T$}
{\bf Input}: A shortest path tree $\overline T$ of $\auxGraph$.\\
{\bf Output}: A tree $T^*$ rooted at $r$ such that every active/inactive vertex of $\auxGraph$ is reachable from $r$ in $T^*$ and every expansion edge is a good expansion edge.
\begin{algorithmic}
\STATE $ \added \gets \emptyset$. \COMMENT{This is the set of edges to be added to $\overline T$ at the end.}
\STATE Let $\Gamma$ be the set of all \textbf{bad expansion edges} with respect to $\overline T$ (cf. Definition \ref{def, good/bad exp}).
\WHILE{$\Gamma\neq \emptyset$}
\STATE pick an arbitrary edge $e=(u,v)\in \Gamma$. Let $w$ be an active vertex such that $(w,u),(v,w)\in E(\auxGraph)$ cf. Lemma \ref{while loop works}. Then
\STATE $\added\gets \added\cup \{(w,u)\}$.
\STATE update $\Gamma$ by removing all the expansion edges incident to $u$ from $\Gamma$. \COMMENT{this makes sure that we add only one edge to $\added$ whose head is $u$}
\ENDWHILE
\STATE $\overline T\cup \added$ is a DAG (cf. Lemma \ref{reachability from r}), so by Claim \ref{valid cut and paste} there exists a subset of edges of $E(\overline T)$ such that its removal makes $\overline T\cup \added$ an arborescence rooted at $r$. Call the resulting arborescence $T^*$.
\RETURN $T^*$.
\end{algorithmic}
\label{modifying bfs alg}
\end{algorithm}

We show Algorithm \ref{modifying bfs alg} works correctly by a series of lemmas.

\begin{lemma}\label{while loop works}
For every bad expansion edge $(u,v)\in E(\overline T)$, there exists an active vertex $w$ such that $(w,u),(v,w)\in E(\auxGraph)$.
\end{lemma}
\begin{proof}
Note that $v$ is a Steiner node, otherwise $(u,v)$ is a good expansion edge with respect to $\overline T$. Let $(u',v)$ be the corresponding edge to $(u,v)$ in $\fs\cup\overline F$. Claim \ref{properties of the endpoints of expansion edge} implies $u'\in B$ for some SCC $B$ of $(V,\fs)$, $v\in A\setminus C_A$ for some $A\in\mathcal{A}_l$, and there is a $\fs$-path from $C_A$ to $B$. Let $w$ be the contraction of $C_A$ in $H$. Note that $u$ is the contraction of $B$ in $H$. Therefore, there is a $\fs$-path from $w$ to $u$ and therefore there is an auxiliary edge $(w,u)$ in $\auxGraph$. The claim follows by noting that there is an edge whose tail is $v$ and enters $C_A$; hence $(v,w)$ is in $\auxGraph$ as well.
\end{proof}

The above claim proves that the while loop in Algorithm \ref{modifying bfs alg} works correctly. Before we prove $\overline T\cup \added$ is a DAG, we need two helper claims.

\begin{claim}\label{added edges are from l to l or l+1 to l}
For any edge $(w,u)\in \added$, we have
\[
\explevel_{\overline T}(u)\leq \explevel_{\overline T}(w)\leq \explevel_{\overline T}(u)+1.
\]
\end{claim}
\begin{proof}
Let $(u,v)$ be the bad expansion edge that caused us to add $(w,u)$ to $\added$ in Algorithm \ref{modifying bfs alg}. By Claim \ref{while loop works} we have $(w,u),(v,w)\in E(\auxGraph)$. Also considering that $\overline T$ is a shortest path tree finishes the proof.
\end{proof}

For the next claim we use the following notation. Note that edges in $\added$ will not form a dipath of length greater than $1$ because the the edges in $\added$ are oriented from an active vertex to an inactive vertex. So for any dipath $P$ in $\overline T\cup\added$ beginning with an edge in $\added$, we write $P = v_1,\added,v_2,\overline T,v_3,\added,...,\overline T,v_k$
where $(v_i,v_{i+1}) \in \added$ for odd $i$ and the subpath $P_{v_i,v_{i+1}}$ uses only edges in $\overline T$ for even $i$.

% We denote by  $P=v_1,\added,v_2,\overline T,v_3,\added,...,\overline T,v_k$ a dipath where from $v_1$ to $v_2$ it uses an edge in $\added$, from $v_2$ to $v_3$ it uses only edges in $\overline T$ and so on and so forth.

\begin{claim}\label{explevel of the endpoin in a path}
Let $k\geq 3$ be odd and let $P=v_1,\added,v_2,\overline T,v_3,\added,...,\overline T,v_k$ be a dipath such that $v_i$ is active for odd $i$ and inactive for even $i$. Then $\explevel_{\overline T}(v_k)\geq \explevel_{\overline T}(v_1)+\frac{k-1}{2}$.
\end{claim}
\begin{proof}
We prove it by induction. Let $k=3$ (i.e., $P=v_1,\added,v_2,\overline T,v_3$). Since $(v_1,v_2)\in \added$ it must be the case that there is a bad expansion edge (with respect to $\overline T$) whose tail is $v_2$; together with the fact that $v_3$ is active and it is in the subtree rooted at $v_2$, we have
\begin{align*}
    \explevel_{\overline T}(v_{3})&\geq \explevel_{\overline T}(v_2)+2\\
    &\geq \explevel_{\overline T}(v_1)+1,
\end{align*}
where the last inequality follows by applying Claim \ref{added edges are from l to l or l+1 to l} to $(v_1,v_2)\in \added$. 

Now suppose the claim holds for $k$ and we prove it for $k+2$. So the dipath is $P=v_1,\added,v_2,\overline T,...,\overline T,v_k,\added,v_{k+1},\overline T,v_{k+2}$.
\begin{align*}
    \explevel_{\overline T}(v_{k+2})&\geq \explevel_{\overline T}(v_{k+1})+2\\
    &\geq (\explevel_{\overline T}(v_k)-1)+2\\
    &\geq \explevel_{\overline T}(v_1)+\frac{k-1}{2}+1\\
    &=\explevel_{\overline T}(v_1)+\frac{k+1}{2},
\end{align*}
where the first inequality follows because there is a bad expansion edge whose tail is $v_{k+1}$ and $v_{k+2}$ is active and it is in the subtree rooted at $v_{k+1}$, the second inequality follows from applying Claim \ref{added edges are from l to l or l+1 to l} to $(v_k,v_{k+1})\in \added$, and the last inequality follows from the induction hypothesis.
\end{proof}

Next we prove the statements after the while loop in Algorithm \ref{modifying bfs alg} works correctly.
\begin{lemma}\label{reachability from r}
After the while loop in Algorithm \ref{modifying bfs alg}, $\overline T\cup \added$ is a DAG.
\end{lemma}
\begin{proof}
Recall that edges in $\added$ will not form a dipath of length greater than $1$; hence, any dicycle in $\overline T\cup \added$ always alternate between a dipath in $\overline T$ and an edge in $\added$. By reordering the alternation, we denote a dicycle in $\overline T\cup \added$ by $C=v_1,\added,v_2,\overline T,v_3,\added,...,\overline T,v_{k-1},\added,v_k,\overline T,v_1$ where $v_i$ is active for odd $i$ and even otherwise. Note that $k$ is even. 

It is easy to see $k\neq 2$. Otherwise we have a dicycle $C=v_1,\added,v_2,\overline T,v_1$ which implies there is a bad expansion edge (with respect to $\overline T$) whose tail is $v_2$ together with the fact that $v_1$ is an active vertex in $\overline T_{v_2}$ we must have $\explevel_{\overline T}(v_1)\geq \explevel_{\overline T}(v_2)+2$ (otherwise all expansion edges whose tail is $v_2$ are good expansion edge). On the other hand, since $(v_1,v_2)\in\added$ by Claim \ref{added edges are from l to l or l+1 to l} we have $\explevel_{\overline T}(v_1)\leq\explevel_{\overline T}(v_2)+1$, a contradiction.

For the sake of contradiction, we assume there is a dicycle $C=v_1,\added,v_2,\overline T,...,\overline T,v_{k-1},\added,v_k,\overline T,v_1$, where $v_i$ is active for odd $i$ and inactive otherwise, furthermore we assume $k\geq 4$. By applying Claim \ref{explevel of the endpoin in a path} to $v_1,\added,v_2,\overline T,...,\overline T,v_{k-1}$, we get $\explevel_{\overline T}(v_1)+1\leq \explevel_{\overline T}(v_{k-1})$, and by applying Claim \ref{added edges are from l to l or l+1 to l} to $(v_{k-1},v_{k})\in\added$ we have $\explevel_{\overline T}(v_{k-1})-1\leq \explevel_{\overline T}(v_k)$. Together, we see $\explevel_{\overline T}(v_1)\leq \explevel_{\overline T}(v_k)$. On the other hand, since $(v_{k-1},v_k)\in\added$ there is a bad expansion edge whose tail is $v_k$, and the fact that $v_1$ is an active vertex in $\overline T_{v_k}$, it must be that $\explevel_{\overline T}(v_k)+2\leq \explevel_{\overline T}(v_1)$ which is a contradiction.
\end{proof}

Next, we show that $\overline T\cup \added$ can be turned into an arborescence by removing a unique subset of edges of $E(\overline T)$. To do so we use the following generic claim.

\begin{claim}\label{valid cut and paste}
Let $T=\big(V(T),E(T)\big)$ be an arborescence, and let $L=\{(u_1,v_1),...,(u_k,v_k)\}$ be a collection of edges such that $u_i,v_i\in V(T)$ and $(u_i,v_i)\notin E(T)$ for all $1\leq i\leq k$. Furthermore, $v_i\neq v_j$ for $i\neq j$. If $T\cup L$ is a DAG, then there is a unique set of edges $B\subseteq E(T)$ of size $k$ such that $(T\cup L)\setminus B$ is an arborescence.
\end{claim}
\begin{proof}
We prove this by induction on the size of $L$. The base case (i.e., when we have one edge in $L$) is easy to see. Suppose it is true when $|L|\leq k$ now we prove it for $|L|=k+1$. Let $L'\subsetneq L$ be a subset of size $k$. Since $T\cup L$ is a DAG so is $T\cup L'$ and hence by induction hypothesis there is a unique $B'\subseteq E(T)$ such that $T':=(T\cup L')\setminus B'$ is an arborescence rooted at $r$. Let $\{e\}=L\setminus L'$, since $T'\cup\{e\}$ is a subgraph of $T\cup L$, we know $T'\cup \{e\}$ is a DAG too and again by induction hypothesis there is an edge $e'\in E(T')$ such that $(T'\cup \{e\})\setminus \{e'\}$ is an arborescence. Since $(T'\cup \{e\})\setminus \{e'\}$ is an arborescence, $e'$ and $e$ must have the same heads (otherwise the head of $e$ has indegree $2$ in $(T'\cup \{e\})\setminus \{e'\}$). The inductive step follows by noticing that $e'$ cannot be in $L$ because otherwise it contradicts the fact that the heads of edges in $L$ are disjoint; hence, $e'\in E(T)$. Let $B:=B'\cup\{e'\}\subseteq E(T)$. Note $|B|=k+1$. Then, we have $(\overline T\cup L)\setminus B=(T'\cup\{e\})\setminus\{e'\}$ which is an arborescence.
\end{proof}

Note that $\overline T\cup \added$ satisfies all the conditions of Claim \ref{valid cut and paste} so the line after the while loop in the algorithm works correctly.

\begin{remark}\label{remark for cut and paste}
The edges in $B$ in Lemma \ref{valid cut and paste} are the edges of $E(T)$ whose head is one of vertices $v_1,...,v_k$. Otherwise some of $v_i$'s have indegree $2$ in $(T\cup L)\setminus B$ which contradicts the fact that $(T\cup L)\setminus B$ is an arborescence.
\end{remark}

Finally, we show that every expansion edge is a good expansion edge (recall Definition \ref{def, good/bad exp}) with respect to $T^*$ to finish the correctness of Algorithm \ref{modifying bfs alg}.
\begin{lemma}\label{good exp edge remains good}
Every expansion edge is a good expansion edge with respect to $T^*$.
\end{lemma}
\begin{proof}
Note that for a bad expansion edge $e=(u,v)$ in $\overline T$ since there is an edge $(w,u)\in \added$ in $T^*$ where $w$ is active, $e$ is a good expansion edge of type 1 with respect to $T^*$.

Next we show that when we are removing edges from $\overline T$ to make $T^*=\overline T\cup\added$ a DAG, we do not make a good expansion edge become bad in $T^*$. By Remark \ref{remark for cut and paste}, we remove $(x,y)\in\overline T$ if and only if there exists an edge in $\added$ whose head is $y$.
\begin{itemize}
    \item[\textbf{case 1.}] If $e=(u,v)$ is a good expansion edge of type 1 in $\overline T$. So there is an active vertex $w$ in $\overline T$ such that the dipath $P_{w,u}$ in $\overline T$ from $w$ to $u$ does not have any expansion edge. Furthermore, if there is an expansion edge whose tail is on $P_{w,u}$, then that expansion edge is of type 1. Hence, there is no edge in $\added$ whose head is in $P_{w,u}$ and so $P_{w,u}$ is in $T^*$ as well and $e$ is a good expansion edge of type 1 in $T^*$.
    \item[\textbf{case 2.}] If $e=(u,v)$ is a good expansion edge of type 2 in $\overline T$. So there is an active vertex $w$ in the subtree of $\overline T$ rooted at $u$ such that the dipath $P_{u,w}$ in $\overline T$ from $u$ to $w$ has at most one expansion edge (it could be that $w=v$). Pick the one that is closest (in terms of edge hops) to $u$. Then all the expansion edges whose tail is on $P_{u,w}$ is of type 2. Therefore, there is no edge in $\added$ whose head is in $P_{u,w}$ and so $P_{u,w}$ is in $T^*$ as well and $e$ is a good expansion edge of type 2 in $T^*$. 
\end{itemize}
\end{proof}

Finally, we can state the proof of the main lemma of this section.

\begin{proof}(of Lemma \ref{the bound on the number of expansion edges})
Consider $T^*$ and assign \textbf{two} tokens to every active vertex. We show that the number of expansion edges is at most the number of tokens to prove the lemma. We do this via the following charging scheme.

~

\noindent
\textbf{Charging scheme:}
At the beginning we label every token \textbf{unused}. We process all the vertices with height $l$. For each expansion edge whose tail has height $l$ we assign an unused token to it and change the label of the assigned token to \textbf{used}. Then we move to height $l-1$ and repeat the process. Fix height $l$. We do the following for every vertex $u$ with this height: if there is no expansion edge whose tail is $u$ then mark $u$ as processed. Otherwise let $(u,v_1),...,(u,v_k)$ be all the expansion edges whose tail is $u$. Note that by definition of type 1 and 2, either (i) all $(u,v_i)$'s are type 1 or (ii) all are type 2. Base on these two cases we do the following:
\begin{itemize}
    \item[\textbf{(i)}] Let $(u,v_1),...,(u,v_k)$ be the expansion edges of type 1. For each $1\leq i\leq k$ there is at least one unused token in $T^*_{v_i}$. Pick one such unused token and assign it to $(u,v_i)$ and change its label to used. Mark $u$ as processed.
    \item[\textbf{(ii)}] Let $(u,v_1),...,(u,v_k)$ be the expansion edges of type 2. For each $1\leq i \leq k$ there is at least one unused token in $T^*_{v_{i}}$. Pick one such unused token and assign it to $(u,v_i)$ and change its label to used. Furthermore, after this there is at least one more unused token in $T^*_{u}$. Mark $u$ as processed.
\end{itemize}
Here we prove by induction on the height $l$, that case (i) and case (ii) works correctly.

Consider the following base case: let $u$ be a vertex and let $(u,v_1),...,(u,v_k)$ be the only expansion edges in $T^*_u$. Then, by Lemma \ref{expansion edges present in arborescence}(b), for every $1\leq i\leq k$ there is an active vertex in $T^*_{v_i}$ and so it has two unused tokens. Therefore, both cases (i) and (ii) work in the base case.

Now consider a vertex $u$ and assume case (i) and case (ii) are correct for all vertices (except $u$) in $T^*_u$ that are the tail of some expansion edges. We show it is correct for $u$ as well.

\textbf{Proof for case (i):} Suppose $u$ falls into case (i). So each $(u,v_i)$ for $1\leq i \leq k$ is of type 1. If there is no expansion edge in $T^*_{v_i}$ then by Lemma \ref{expansion edges present in arborescence}(b) there is an active vertex in $T^*_{v_i}$ and has two unused tokens. So now assume there is an expansion edge in $T^*_{v_i}$ and pick the one $f_i=(x_i,y_i)$ whose tail is closest to $v_i$ (break the ties arbitrarily). If $f_i$ is of type 2, then by induction hypothesis $T^*_{x_i}$ has one unused token (when we processed $x_i$) and since by the choice of $f_i$ there is no expansion edge on the dipath $P_{v_i,x_i}$ in $T^*$; hence this token is unused at this iteration as well. If $f_i$ is of type 1, then there is an active vertex $z$ on $P_{v_i,x_i}$ and has two tokens. Again we note that the tokens of $z$ are unused since there is no expansion edge on $P_{v_i,z}$.

So we proved for each $(u,v_i)$ where $1\leq i\leq k$ there is at least one unused token in $T^*_{v_i}$, as desired.

\textbf{Proof for case (ii):} Suppose $u$ falls into case (ii). So each $(u,v_i)$ for $1\leq i \leq k$ is of type 2. With the exact same argument as in case (i), we can show that for each $1\leq i\leq k$ there is (at least) one unused token in $T^*_{v_i}$. So we just need to show an extra unused token in $T^*_u$.

Since $(u,v_i)$'s are of type 2, there must be an active vertex $w$ such that $\explevel_{T^*}(u)\leq \explevel_{T^*}(w)\leq \explevel_{T^*}(u)+1$. Pick such $w$ with smallest $\explevel$. If $\explevel_{T^*}(w)=\explevel_{T^*}(u)$ then $w$ has two tokens and these tokens are different than the ones in $T^*_{v_i}$ because $w$ is not in $T^*_{v_i}$'s. Furthermore, the tokens of $w$ are unused because there is no expansion edge on the dipath $P_{u,w}$ in $T^*$.

So let us assume $\explevel_{T^*}(w)=\explevel_{T^*}(u)+1$. There are two cases to consider:
\begin{itemize}
    \item $w$ is in $T^*_{v_j}$ for some $1\leq j\leq k$. Note that there is no expansion edge on $P_{v_j,w}$. Therefore, among the two tokens of $w$, one could be assigned to $(u,v_j)$ as before and the other one will be unused when we are processing $u$ so this would be the extra unused token we wanted.
    \item $w$ is not in $T^*_{v_j}$ for any $1\leq j\leq k$. So there is one expansion edge $(x,y)$ on $P_{u,w}$. By the choice of $w$ (with smallest $\explevel$) together with the fact that all $(u,v_i)$'s are of type 2, implies $(x,y)$ must be of type 2. Therefore, $x$ has one unused token when $x$ was processed. Since there is no expansion edge on $P_{u,x}$, this token is unused at this iteration as well. Finally, since $x$ is not in $T^*_{v_i}$ for $1\leq i\leq k$ this unused token is the extra token, as desired.
\end{itemize}
\end{proof}

\subsection{Putting everything together}\label{putting everything together}
Fix an iteration $l$. We use Lemmas \ref{bound for the number of killer edges} \& \ref{the bound on the number of expansion edges} and the properties of graph $G$ to bound $\sum\limits_{A\in\mathcal{A}_l}|\Delta^{l}_{\killer}(A)\cup\Delta^{l}_{\expansion}(A)|$. Consider an active moat $A$ and its SCC $C_A$. We show there is at most one killer/expansion edge that enters $C_A$. So the remaining killer/expansion edges must enter some Steiner node in $A\setminus C_A$. We use this fact later.

\begin{claim}\label{claim needed for contraction}
Fix an iteration $l$ and an active moat $A\in\mathcal{A}_l$. There is at most one edge in $\Delta^{l}_{\killer}(A)\cup\Delta^{l}_{\expansion}(A)$ whose head is in $C_A$.
\end{claim}
\begin{proof}
Suppose there are two edges $e$ and $f$ in $\Delta^{l}_{\killer}(A)\cup\Delta^{l}_{\expansion}(A)$ that enter $C_A$. Since $e$ and $f$ are bought later than all the edges in $C_A$, we should have pruned one of $e$ or $f$ in the deletion phase.
\end{proof}

Consider the graph $\fs\cup \overline F$. Remove all vertices that are not in an active moat at this iteration. For each active moat $A$, remove all Steiner nodes in $A\setminus C_A$ that are not the head of any edge in $\overline F^l_{\killer}\cup\overline F^l_{\expansion}$. Then, for each $A\in\mathcal{A}_l$ contract $C_A$ to a single vertex and call the contracted vertex by $C_A$. Finally, if there are parallel edges, arbitrarily keep one of them and remove the rest\footnote{Note that all the parallel edges are antenna edges and so removing them does not affect the quantity $\sum\limits_{A\in\mathcal{A}_l}|\Delta^{l}_{\killer}(A)\cup\Delta^{l}_{\expansion}(A)|$ we are trying to bound.}. Call the resulting graph $G'$.

Now we relate the sum we are interested in to bound with the sum of the indegree of vertices in $G'$.

\begin{claim}\label{counting total indegree of active moats}
For each active moat $A\in\mathcal{A}_l$, we have
\begin{equation}\label{relation between indegrees}
|\Delta^l_{\killer}(A)\cup\Delta^l_{\expansion}(A)|\leq|\delta^{in}_{G'}(C_A)|+1.
\end{equation}
\end{claim}
\begin{proof}
Consider an active moat $A$ and let $v$ be a Steiner node in $A\setminus C_A$. First note that the indegree of vertices in $\overline F$ is at most $1$ therefore there is at most one edge $e\in\overline F^l_{\killer}\cup\overline F^l_{\expansion}$ that enters $v$. Secondly, we note that by Lemma \ref{remark on active moats} there is at least one edge in $\fs$ from $v$ to $C_A$ and we kept one such edge in $G'$; so the contribution of $e$ to the LHS of (\ref{relation between indegrees}) is accounted for in the RHS. Finally, by Claim \ref{claim needed for contraction} at most one killer/expansion edge enters $C_A$ and the contribution of this edge is accounted for by the plus one in the RHS.
\end{proof}

Next, using Lemmas \ref{bound for the number of killer edges} \& \ref{the bound on the number of expansion edges} we bound the number of vertices in $G'$.
\begin{claim}\label{number of vertices in G'}
Fix an iteration $l$. Then, $|V(G')|\leq 4\cdot|\mathcal{A}_l|$.
\end{claim}
\begin{proof}
The set $V(G')$ is consist of $C_A$'s for some active moat $A$ and bunch of Steiner nodes. Note that we kept a Steiner node $s$ if there is (exactly) one edge in $\overline F^l_{\killer}\cup\overline F^l_{\expansion}$ that enters $s$. Therefore, $|V(G')|$ is at most $|\mathcal{A}_l|+|\overline F^l_{\killer}|+|\overline F^l_{\expansion}|$. The bound follows from Lemmas \ref{bound for the number of killer edges} \& \ref{the bound on the number of expansion edges}.
\end{proof}

Finally, we prove Theorems \ref{main thm minor free graphs} \& \ref{main thm planar graphs}.

\begin{proof}(of Theorem \ref{main thm minor free graphs})
Since $G$ is $K_r$-minor free so does $G'$. So we can write
\begin{equation}\label{number of killer and exp finally}
    \begin{aligned}
        \sum\limits_{A\in\mathcal{A}_l}\big|\Delta^{l}_{\killer}(A)\cup\Delta^{l}_{\expansion}(A)\big|&\leq \sum\limits_{A\in\mathcal{A}_l}\big(|\delta^{in}_{G'}(C_A)|+1\big)\\
        &=|E(G')|+|\mathcal{A}_l|\\
        &\leq O(r\cdot\sqrt{\log r})\cdot 4\cdot|\mathcal{A}_l|+|\mathcal{A}_l|\\
        &=O(r\cdot\sqrt{\log r})|\mathcal{A}_l|,
    \end{aligned}
\end{equation}
where the inequality follows from Claim \ref{counting total indegree of active moats} and the second inequality follows from Claim \ref{number of vertices in G'} together with Theorem \ref{size of minor free graphs}.

Next we show (\ref{ineqaulity primal-dual analysis}) holds for $\alpha=O(r\cdot\sqrt{\log r})$.
\begin{equation*}
    \begin{aligned}
        \sum\limits_{A\in\mathcal{A}_l}|\Delta^{l}(A)|&=\sum\limits_{A\in\mathcal{A}_l}|\Delta^{l}_{\killer}(A)\cup\Delta^{l}_{\expansion}(A)|+\sum\limits_{A\in\mathcal{A}_l}|\Delta^{l}_{\antenna}(A)|\\
        &\leq O(r\cdot\sqrt{\log r})|\mathcal{A}_l|+|\mathcal{A}_l|\\
        &=O(r\cdot\sqrt{\log r})|\mathcal{A}_l|,
    \end{aligned}
\end{equation*}
where inequality follows from inequality (\ref{number of killer and exp finally}) and Lemma \ref{number of antenna edges}.

As we discussed at the beginning of Section \ref{the analysis} that if (\ref{ineqaulity primal-dual analysis}) holds for $\alpha$ then we have a $(2\cdot\alpha)$-approximation algorithm. Hence, Algorithm \ref{our alg} is an $O(r\cdot\sqrt{\log r})$-approximation for DST on quasi-bipartite, $K_r$-minor free graphs.
\end{proof}

\begin{proof}(of Theorem \ref{main thm planar graphs})
The proof of Theorem \ref{main thm planar graphs} is exactly the same as proof of Theorem \ref{main thm minor free graphs} except instead of $O(r\cdot \sqrt{\log r})$ in (\ref{number of killer and exp finally}) we have $2$ because $G'$ is a bipartite planar graph, see Lemma \ref{size of bipartite planar}. Now we can write
\begin{equation*}
        \sum\limits_{A\in\mathcal{A}_l}\big|\Delta^{l}_{\killer}(A)\cup\Delta^{l}_{\expansion}(A)\big|\leq 9\cdot|\mathcal{A}_l|,
\end{equation*}
and
\begin{equation*}
        \sum\limits_{A\in\mathcal{A}_l}|\Delta^{l}(A)|\leq 10\cdot|\mathcal{A}_l|.
\end{equation*}
Therefore, (\ref{ineqaulity primal-dual analysis}) holds for $\alpha=10$ and hence we have a $20$-approximation algorithm, as desired.
\end{proof}

\section{NP-hardness}\label{np hardness}
In this section we prove Theorem \ref{hardness}. We reduce from the $\np$-complete problem \CNC (CVC) on planar graphs. Here, we are given a planar graph $G=(V,E)$ and a positive integer $k$. The goal is to decide if there is a vertex cover $S\subseteq V$ such that $|S|\leq k$ and $G[S]$ (the induced subgraph on $S$) is connected. This problem is shown to be $\np$-complete, see Lemma 2 in \cite{garey1977rectilinear}.

Our reduction from CVC on planar graphs to \ST on quasi-bipartite planar graphs is similar to the reduction showing \ST problem is $\np$-hard on general graphs from \cite{karp1972reducibility}. Let $(G=(V,E),~k)$ be an instance of CVC where $G$ is planar. Subdivide every edge $e\in E$ by a terminal vertex $x_e$ and call the resulting graph $G'$, which is also planar. Let $X:=\{x_e:~\forall e\in E\}$ be the set of terminals and $V(G)$ is the set of Steiner nodes in $G'$.

\begin{lemma}
$G'$ has a Steiner tree of size $k+|E(G)|-1$ if and only if $G$ has a connected vertex cover of size $k$.
\end{lemma}
\begin{proof}
Suppose $G$ has a connected vertex cover $S$ of size $k$. Then, $G'[S\cup X]$ is connected and therefore it has a spanning tree $T$ where $|E(T)|=|S\cup X|-1=|E(G)|+k-1$.

Now let $T'$ be a tree that spans $X$ in $G'$ and $|E(T')|=|E(G)|+k-1$. Since $|X|=|E(G)|$, we have $|V(T')\setminus X|=k$. Define $S:=V(T')\setminus X$; we show that $S$ is a connected vertex cover for $G$. The fact that it is a vertex cover is clear because for every edge $e\in E(G)$ at least one of its endpoint is in $V(T')\setminus X$. Consider $u,v\in S$. Since $u,v\in V(T')$, there is a path $P=u,x_{e_1},w_1,x_{e_2},w_2,...,w_{l-1},x_{e_l},v$ in $T'$. Note that the path $(u,w_1),(w_1,w_2),...,(w_{l-1},v)$ is in $G[S]$. So we showed $G[S]$ is a connected subgraph of $G$ and $|S|=k$, as desired.
\end{proof}
This completes the proof of Theorem \ref{hardness}.
%So we showed a polynomial time reduction from \CNC to a \ST problem on bipartite planar graphs where the terminals are in one side and the Steiner nodes are on the other side which proves Theorem \ref{hardness}.

\bibliographystyle{alpha}
\bibliography{references}

\end{document}